\documentclass{article}
\usepackage{graphicx}
\usepackage{amsfonts,amsmath,amssymb,amsthm,bbm}
\usepackage{braket}
\usepackage{hyperref}
\hypersetup{colorlinks=true,linkcolor=blue,citecolor=blue}
\usepackage{cite}
\usepackage{xcolor}
\usepackage[hmargin=1.1in, vmargin=1.3in]{geometry}
\usepackage{mathtools}
\usepackage{footnote}
\makesavenoteenv{enumerate}
\makesavenoteenv{itemize}
\usepackage{authblk}

\usepackage{algorithm}
\usepackage{algorithmic}

\usepackage[normalem]{ulem}
\usepackage{cancel}

\newtheorem{theorem}{Theorem}
\newtheorem{corollary}[theorem]{Corollary}
\newtheorem{lemma}[theorem]{Lemma}
\newtheorem{proposition}[theorem]{Proposition}
\newtheorem{definition}[theorem]{Definition}

\usepackage{titlesec}
\titleformat{\section}[block]
  {\normalfont\bfseries\large\filcenter}
  {\thesection.}{1em}{}
\titleformat{\subsection}[block]
  {\normalfont\bfseries\normalsize\filcenter}
  {\thesubsection.}{1em}{}
\titleformat{\subsubsection}[block]
  {\normalfont\normalfont\filcenter}
  {\thesubsubsection.}{1em}{}
\titlespacing*{\section}
  {0pt}{3em}{2em}
\titlespacing*{\subsection}
  {0pt}{3em}{2em}
\titlespacing*{\subsubsection}
  {0pt}{3em}{2em}
\renewcommand{\thesection}{\Roman{section}}

\usepackage{tocloft}
\cftsetindents{section}{0em}{1.6em}
\cftsetindents{subsection}{1.5em}{2.3em}

\usepackage{tikz}
\usetikzlibrary{decorations.markings}
\usetikzlibrary{arrows.meta}
\tikzset{midarrow/.style={postaction={decorate, decoration={markings, mark=at position 0.6 with {\arrow{latex}}}}}}

\DeclareMathOperator*{\E}{\mathbbm{E}}
\DeclareMathOperator{\Tr}{\mathrm{Tr}}
\DeclareMathOperator{\PP}{\mathbb{P}}
\newcommand{\dd}{\mathrm{d}}
\newcommand{\OO}{\mathcal{O}}

\newcommand{\ii}{\mathrm{i}}
\newcommand{\calA}{\mathcal{A}}
\newcommand{\calF}{\mathcal{F}}
\newcommand{\calW}{\mathcal{W}}

\DeclareMathOperator{\calP}{\mathcal{P}}
\newcommand{\frakd}{\mathfrak{d}}
\newcommand{\onetoone}{{\ensuremath{1\text{-}1}}}
\newcommand{\inftoinf}{{\ensuremath{\infty\text{-}\infty}}}

\newcommand{\htilde}{\widetilde{h}}
\newcommand{\poly}{\mathrm{poly}}
\newcommand{\supp}{\mathrm{supp}}
\DeclareMathOperator{\majsupp}{\operatorname{supp}}
\DeclareMathOperator{\vecrm}{\mathrm{vec}}
\newcommand{\majone}{{\ensuremath{\mathrm{maj},1}}}

\newcommand{\istrue}{\boldsymbol{1}}
\newcommand{\sgn}{\mathrm{sgn}}
\newcommand{\frakI}{\mathfrak{I}}

\newcommand{\Or}{\mathcal{O}}

\title{Efficient Classical Simulation \\ of Weakly Interacting Fermion Dynamics}

\author[1,4]{Chu Zhao}
\author[1,3,4]{Iman Marvian}
\author[1,2,4]{Yu Tong}

\affil[1]{Department of Electrical and Computer Engineering, Duke University, Durham, NC 27708, USA}
\affil[2]{Department of Mathematics, Duke University, Durham, NC 27708, USA}
\affil[3]{Department of Physics, Duke University, Durham, NC 27708, USA}
\affil[4]{Duke Quantum Center, Duke University, Durham, NC 27701, USA}

\date{\today}

\begin{document}

\maketitle

\begin{abstract}
    We consider the task of simulating the real-time dynamics of weakly interacting fermionic systems. In particular, we focus on computing the expectation value of a local observable $A$ at time $t$. By analyzing the convergence of the perturbative expansion in the interaction strength $\lambda$ for the Heisenberg-picture observable, we propose a polynomial-time algorithm for estimating this expectation value in the weakly interacting regime $\lambda |t|^{2D+1}=\mathcal{O}(1)$, when the Hamiltonian is geometrically local on a $D$-dimensional lattice. Importantly, this condition is independent of the system size.
    If the goal is instead to approximate the time-evolved observable in normalized Frobenius norm, we extend the convergence regime to $\lambda |t|=\mathcal{O}(1)$ with quasi-polynomial runtime.
    When the non-interacting part exhibits Anderson localization, our polynomial-time algorithm can be extended up to $\lambda |t|=\mathcal{O}(1)$, modulo polylogarithmic factors.
    Our algorithm brings together ideas from continuous-time QMC, diagrammatic QMC, and Majorana Propagation, but with a new Heisenberg-picture operator-growth analysis that makes the sampling complexity rigorously controllable. This leads to provably efficient classical algorithms in regimes where the interaction is weak enough that the sampling variance remains bounded independently of system size.
    Together, these results identify broad regimes in which weak interactions, locality, and localization can be leveraged to make real-time fermionic dynamics classically tractable.
\end{abstract}

\section{Introduction}\label{sec:intro}

Simulating the real-time dynamics of interacting fermions is a central problem in quantum many-body physics, with direct relevance to condensed matter systems, quantum chemistry, materials science, and nonequilibrium statistical mechanics. Although generic interacting fermionic systems are believed to be intractable for classical computers \cite{Bao2015universal,childs2013universal}, this does not preclude efficient algorithms for structured regimes of physical interest. One particularly natural regime is that of fermionic Hamiltonians with a dominant quadratic component and weak quartic interactions, which interpolate between exactly solvable free-fermion dynamics and fully interacting many-body behavior.

In this work, we consider weakly interacting fermionic systems described by a Hamiltonian of the following form:
\begin{equation}\label{eq:interacting_fermions_intro}
    H := H^0 + \lambda V := \sum_{i,j=1}^{2N} h_{ij} \gamma_i \gamma_j + \lambda \sum_{i,j,k,l=1}^{2N} v_{ijkl} \gamma_{i}\gamma_{j}\gamma_{k}\gamma_{l},
\end{equation}
where $\lambda > 0$ is a perturbation parameter.
The operators $\gamma_i, \gamma_j, \gamma_k$, and $\gamma_l$ are Majorana operators.
$h$ is a purely imaginary antisymmetric matrix, satisfying $h_{ij} = - h_{ji}$, which makes it Hermitian.
$v$ is a real rank-4 tensor satisfying $v_{\pi(i)\pi(j)\pi(k)\pi(l)} = \sgn(\pi) v_{ijkl}$ for any permutation $\pi\in S_4$. These assumptions ensure that $H^0$ and $V$ are Hermitian. We assume that the parameters are bounded: $|h_{ij}| \le 1, |v_{ijkl}| \le 1$.
We will consider two settings: (1) $H$ is a bounded-degree local Hamiltonian, i.e., each Hamiltonian term is supported on at most a constant number of Majorana modes and each Majorana mode is involved in at most a constant number of terms; and (2) $H$ involves short-range interactions on a $D$-dimensional lattice. The latter setting is more restrictive than the former.

Despite the restriction to small $\lambda$, these models have rich phenomenology that has been the focus of many studies. In particular, the emergence of a classical kinetic equation in the $\lambda\to 0$ limit for time $t\lesssim \lambda^{-2}$ has been extensively studied, typically in a translation-invariant setting \cite{vanHove1954quantum,hugenholtz1983derivation,HoLandau1997fermi,benedetto2008n,lukkarinen2009not}. Weakly interacting fermionic systems exhibit interesting prethermal behaviors \cite{stark2013kinetic,lukkarinen2017kinetic}, and have been an important tool for studying transport because they enable controlled approximation while already producing behavior qualitatively distinct from free fermions \cite{KielyMueller2021}. Because ultracold fermions in optical lattices provide experimental realizations of Hubbard-type models, the weakly interacting models studied in this work are useful as benchmark tools to validate experimental results \cite{Esslinger2010fermi,Jordens2008mott,Schneider2008metallic,Cheuk2016observation,mazurenko2017cold,bakr2025microscopy}.

Due to the importance of this problem, many numerical algorithms have been developed for this task. Recently, a pioneering work by Facelli et al. \cite{facelli2026fastconvergencemajoranapropagation} applied the Majorana Propagation algorithm, originally designed for simulating fermionic circuits \cite{miller2025simulation,DAnna2025majorana}, to this problem, to obtain an algorithm with provable convergence and quasi-polynomial runtime, going beyond numerical heuristics. More specifically, for a weakly interacting fermionic Hamiltonian as described in \eqref{eq:interacting_fermions_intro} that is bounded-degree local, \cite{facelli2026fastconvergencemajoranapropagation} proposes an algorithm that can approximate a local observable to
precision $\epsilon$ in normalized Frobenius norm with runtime $N^{\Or(\log(t/\epsilon))}$ as long as $t= \Or(\log(1/\lambda))$ (see \cite[Theorem 1.2]{facelli2026fastconvergencemajoranapropagation}).

In this work, we improve this result in two ways by more refined analysis and improved algorithm design. In particular, we will
(1) improve the convergence analysis to cover the parameter regime $t=\Or(1/\lambda)$ for bounded-degree local Hamiltonians, thereby \emph{exponentially extending the time scale} that can be simulated using the method in \cite{facelli2026fastconvergencemajoranapropagation}; 
(2) propose an algorithm that runs in time $\Tilde{\Or}(N^3/\epsilon^2)$, where $\tilde{O}(\cdot)$ hides logarithmic factors in $1/\epsilon$, as long as $t=\Or((1/\lambda)^{\frac{1}{2D+1}})$ on a $D$-dimensional lattice, thereby \emph{superpolynomially improving the runtime} compared to the state of the art.\footnote{We note that for (2) the more precise condition is $\lambda|t|(1+|t|)^{2D}=\Or(1)$ (Theorem~\ref{thm:convergence_maj_one}), but we can simplify it to $t=\Or((1/\lambda)^{\frac{1}{2D+1}})$ when considering sufficiently large $t>0$.} We will discuss these results in more detail below.

\medskip
\noindent \textbf{Main results.}
In this paper, we present a simulation algorithm that computes the quantity
\[
\Tr(\rho_0e^{iHt } A e^{-iHt })
\]
where $\rho_0$ is any state for which the expectation value of Majorana monomials can be efficiently computed, such as a
fermionic Gaussian state, i.e., the thermal state of a fermionic quadratic Hamiltonian, including pure states obtained as limiting cases. 
$A$ is a local observable involving a constant number of Majorana modes.
$H$ is the Hamiltonian for a weakly interacting fermionic system, as defined in \eqref{eq:interacting_fermions_intro}.
In our simulation algorithm, we exploit the fact that the free-fermion dynamics is classically easy to compute.
We first define $\widetilde{A}(t) = e^{-iH^0 t} e^{iHt} A e^{-iHt} e^{iH^0 t}$ to isolate the dynamics due to the interacting part $\lambda V$.
Our algorithm is based on a perturbative expansion
of $\widetilde{A}(t)$ in $\lambda$.
If this expansion converges exponentially, we can truncate the series at a suitable order to obtain an approximation with the desired accuracy.

We obtain the perturbative expansion of $\widetilde{A}(t)$ in $\lambda$ by computing the Dyson series for $\widetilde{A}(t)$. Note that
\begin{equation*}
    \frac{\dd }{\dd t } \widetilde{A}(t) = i \lambda [ \tau^0_{-t} (V), \widetilde{A}(t)],
\end{equation*}
where we denote the free evolution in the Heisenberg picture under $H^0$ by
\begin{equation}
    \tau^0_t = e^{iH^0 t} ( \, \cdot \,) e^{-iH^0 t}.
\end{equation}
We therefore obtain the perturbative expansion
\begin{subequations} \label{eq:series_expansion}
\begin{equation}
    \label{eq:series_expansion_sum}
    \widetilde{A}(t) = \sum_{k=0}^\infty \widetilde{A}^{(k)}(t),
\end{equation}
where
\begin{equation} \label{eq:kth_order}
    \widetilde{A}^{(k)}(t) = (i\lambda)^k \int_0^{t} \dd t_k \int_0^{t_{k}} \dd t_{k-1} \cdots \int_0^{t_{2}} \dd t_1 \, \Big[ \tau^0_{-t_k} (V), \big[\tau^0_{-t_{k-1}}(V), \cdots, [\tau^0_{-t_1}(V), A]\big] \Big].
\end{equation}
\end{subequations}

Our first result shows that the series \eqref{eq:series_expansion} converges exponentially in the normalized Frobenius norm  when $t=\Or(1/\lambda)$, provided that the Hamiltonian is bounded-degree local.\footnote{We will use ``Frobenius norm'' and ``normalized Frobenius norm'' interchangeably throughout the paper.} We assume that each Majorana mode appears in at most $\frakd$ Hamiltonian terms and each Hamiltonian term involves at most $\kappa$ Majorana modes. We call such a Hamiltonian $(\frakd, \kappa)$-bounded-degree local (see Definition \ref{def:bounded_deg_local}). $\frakd$ and $\kappa$ can generally be regarded as constants that are independent of the system size.
For our Hamiltonian \eqref{eq:interacting_fermions_intro}, $\kappa=4$.
This locality condition controls operator spreading in the Heisenberg picture and thereby ensures that the growth of the Frobenius norm is independent of the system size.

\begin{theorem}
\label{thm:Fnorm_convergence_intro}
    Consider an
    interacting fermionic system
    \eqref{eq:interacting_fermions_intro}
    with $(\frakd, 4)$-bounded-degree locality (Definition~\ref{def:bounded_deg_local}).
    For any initial observable $A$ and $k\ge 1$,
    \begin{equation*}
        \| \widetilde{A}^{(k)}(t) \|_F \le (2 e^2\lambda |t| \frakd)^k e^{d_0-2} \|A\|_F,
    \end{equation*} where $d_0 = \deg(A)$.
    The series \eqref{eq:series_expansion} converges exponentially when $2 e^2\lambda |t| \frakd < 1$.
\end{theorem}

Based on this result, the integral \eqref{eq:kth_order}, and hence $\widetilde{A}(t)$ and $A(t):= e^{iHt} A e^{-iHt}$, can be evaluated by truncating the Dyson series at order $K = \Or(\log (1/\epsilon))$.
Viewing the integrand as a polynomial in Majorana operators of degree $\Or(k)$, each order can be estimated by Monte Carlo integration in the degree-$\Or(K)$ Majorana subspace. This gives the quasi-polynomial runtime stated in Corollary \ref{cor:quasi-polynomial-time} below.

We remark that Theorem \ref{thm:Fnorm_convergence_intro} and Corollary \ref{cor:quasi-polynomial-time} can be generalized to Hamiltonians with an arbitrary quadratic part, whose coefficients need not be bounded by \(1\) and which need not be bounded-degree local. As we shall see in the proof in Section \ref{sec:convergence_in_F_norm}, the free evolution does not contribute to the growth of the Frobenius norm due to unitarity.

\begin{corollary} \label{cor:quasi-polynomial-time}
    Consider an
    interacting fermionic system
    \eqref{eq:interacting_fermions_intro}
    with $(\frakd, 4)$-bounded-degree locality (Definition~\ref{def:bounded_deg_local}).
    Let $A$ be an initial observable given as a list of Majorana monomials.
    When $|t| < 0.99/(2e^2 \lambda \frakd)$, there is an
    \[
        {\Or}\left( e^{2d_0}\|A\|_F^2 K^2 N^{\Or(K)} \log \frac{K}{\delta}\right)
    \]
    -time algorithm, where $d_0 = \deg(A)$ and
    $K = \Or(d_0 + \log \|A\|_F + \log \frac{1}{\epsilon})$,
    that computes $\widetilde{A}(t)$ and $A(t):= e^{iHt} A e^{-iHt}$ to within $\Or(\epsilon)$ additive error in the normalized Frobenius norm with success probability $1-\delta$.

\end{corollary}

This result exponentially extends the $|t| = \Or(\log (1/\lambda))$ regime of \cite{facelli2026fastconvergencemajoranapropagation}.
Here we regard $d_0, \|A\|_F$ as constants.
The above algorithm runs in time $\Tilde{\Or}( N^{\Or(\log ( 1 / \epsilon ) ) } )$, where $\Tilde{\Or}(\cdot)$ hides logarithmic factors in $\delta$.
Unlike the $N^{\Or(\log(|t|/\epsilon))}$ runtime of \cite{facelli2026fastconvergencemajoranapropagation}, our runtime does not depend on $|t|$, as long as $|t|$ is below the threshold,
which is preferable when $\lambda$ is small and $|t|$ is large.
Compared to the deterministic algorithm in \cite{facelli2026fastconvergencemajoranapropagation}, our algorithm is randomized.
Our approach to approximating the time-evolved observable
also differs from that used in \cite{facelli2026fastconvergencemajoranapropagation}. Both approaches obtain a low-degree approximation, but we approximate the time evolution using a truncated Dyson series, whereas \cite{facelli2026fastconvergencemajoranapropagation} does so using Trotterization.
A key difference between our analysis and that of \cite{facelli2026fastconvergencemajoranapropagation} is that we use the interaction picture to keep $H^0$ in the exponential.
This allows us to track the operator spreading caused by $e^{-iH^0 t}$ separately from that caused by $\lambda V$ and to benefit from the unitarity of the free evolution. The proof of Theorem \ref{thm:Fnorm_convergence_intro} and Corollary \ref{cor:quasi-polynomial-time} is presented in Section \ref{sec:convergence_in_F_norm}.

The convergence and precision guarantees stated above and in \cite{facelli2026fastconvergencemajoranapropagation} are all measured in the normalized Frobenius norm, which provides a notion of the average error of the expectation value approximation over random input states. The above result therefore does not guarantee the worst-case error for all input states, which requires an operator-norm error bound.
In the worst case, the operator norm and the normalized Frobenius norm differ by a dimension-dependent factor.
As we will see later, when the Hamiltonian is geometrically local, the Dyson series truncation performs significantly better than predicted by Theorem \ref{thm:Fnorm_convergence_intro}. In particular, we quantify the error using the \textit{Majorana $1$-norm}, which we define as the $\ell_1$-norm of the vector coefficients obtained by expanding the operator in the Majorana monomial basis (see Definition~\ref{def:norm_notation}). By definition, the Majorana $1$-norm upper bounds operator norm. On the other hand, Majorana $1$-norm is useful for designing a randomized sampling algorithm, introduced shortly, that reduces the quasi-polynomial runtime to polynomial time. We first establish the following exponential convergence of the series \eqref{eq:series_expansion} in Majorana $1$-norm when $t=\Or((1/\lambda)^{\frac{1}{2D+1}})$ for locally interacting fermions on a $D$-dimensional lattice. We assume that all Hamiltonian terms have finite range $r_0$ and each Majorana mode appears in at most $\frakd$ Hamiltonian terms. We call such a Hamiltonian $(r_0, \frakd)$-geometrically local (see Definition \ref{def:geo_local}).

\begin{theorem}\label{thm:convergence_maj_one}
    Consider an
    interacting fermionic system
    \eqref{eq:interacting_fermions_intro}
    on a $D$-dimensional lattice with a $(r_0, \frakd)$-geometrically local Hamiltonian (Definition~\ref{def:geo_local}). Let
    \begin{equation*}
        W_t = 2 \frakd \left( 1 + C_D  \lceil 4e \frakd |t| r_0 \rceil^{D/2}  \right)^4,
    \end{equation*} where $C_D > 0$ is a constant that depends only on $D$. Then, for any initial observable $A$ and $k\ge 1$, we have
    \begin{equation*}
        \| \widetilde{A}^{(k)}(t) \|_{\majone} \le (e^2 \lambda |t| W_t )^k e^{d_0-2} \|A\|_\majone,
    \end{equation*}
    where $d_0 = \deg(A)$.
    The series \eqref{eq:series_expansion} converges exponentially when $e^2 \lambda |t|W_t < 1$.
\end{theorem}

Again, here $r_0$, $D$, $\|A\|_{\mathrm{maj},1}$, and $d_0$
can be regarded as constants that are independent of the system size.
This convergence theorem
also follows from an analysis of operator spreading in the Heisenberg picture.
In particular, for geometrically local Hamiltonians on a lattice, time evolution spreads operators at a finite velocity, leading to a controlled growth of the Majorana $1$-norm. We prove a quantitative statement describing this in Lemma~\ref{lem:eiht_bnd} for degree-$1$ Majorana monomials under free evolution.
The proof of the theorem is given in Section \ref{sec:convergence_in_one_norm}.

Based on this result, we propose an efficient randomized algorithm for computing \(\Tr(\rho_0 e^{iHt}Ae^{-iHt})\). The algorithm samples the perturbative expansion \eqref{eq:series_expansion} directly in the Majorana basis.
For each order up to the truncation order, each commutator produces a linear combination of Majorana monomials; rather than retaining the full combination, we randomly sample one monomial with probability proportional to the magnitude of its coefficient and keep track of the corresponding sampling weight.
This produces unbiased samples for $\widetilde{A}^{(k)}(t)$.
The Majorana $1$-norm controls the sampling weights and hence the sampling variance, which remains bounded in the same parameter regime as in Theorem \ref{thm:convergence_maj_one}. The algorithm is presented in detail in Section \ref{sec:algorithm}.

\begin{theorem}[Theorem~\ref{thm:sample_complexity_l1}, informal]
    Consider
    locally interacting fermions on a $D$-dimensional lattice.
    Let $N$ be the number of fermionic modes,
    $A=\sum_{i=1}^M a_i O_i$ be an observable given as a list of Majorana monomials $\{O_i\}_i$ with coefficients $\{a_i\}_i$ satisfying $ |a_i|\le 1$, and $\rho_0$ be a Gaussian state.
    Then there exists a constant
    $c = \Or(1)$ such that
    when $\lambda |t| (1+|t|)^{2D} < c$,
    there is a sampling algorithm that computes $\Tr(\rho_0\, e^{iHt } A e^{-iHt })$ to within additive error $\epsilon$ with probability at least $1-\delta$ in time
    $\widetilde{\OO}\left(\epsilon^{-2}M^3 e^{2d_0} N^3 \right)$, where $d_0 = \deg(A)$ and $\widetilde{\OO}(\cdot)$ hides logarithmic factors in $M, 1/\epsilon$ and $1/\delta$.
\end{theorem}

As an example, we consider the Anderson model \cite{stolz2011introanderson}. When disorder is added to the system, quantum transport is greatly suppressed. This property is called \textit{dynamical localization}. More precisely, it means that local single-particle wave functions remain localized in space uniformly for all times.
Therefore, we expect that in the Anderson model,
operator spreading is more tightly bounded, which may lead to faster convergence and hence a higher time threshold.
Indeed, we show a faster convergence in the following corollary. To translate dynamical localization to a uniform probability bound on the transport, we restrict the Hamiltonian to a subregion of radius $R = \OO(|t| + \log \frac{1}{\epsilon})$
by sacrificing $\epsilon$ precision in the time dynamics simulation and apply a union bound to this subregion.
We show that the formal statement of Theorem~\ref{thm:convergence_maj_one} remains valid, while the quantity $W_t$, which captures the overall effect of operator spreading, can be improved to $\poly \log ( 1 + |t| )$, hence extending the parameter regime of our algorithm.

\begin{corollary}[Corollary~\ref{cor:samplecomplexity_anderson}, informal]
    Consider locally
    interacting fermions on a $D$-dimensional lattice with disorder large enough to cause the spectrum to localize. Let $N$ be the number of fermionic modes,
    $A=\sum_{i=1}^M a_i O_i$ be an observable given as a list of Majorana monomials $\{O_i\}_i$ with coefficients $\{a_i\}_i$ satisfying $|a_i|\le 1$, and $\rho_0$ be a Gaussian state.
    Then when $\lambda|t| \, \operatorname{polylog} (1/\epsilon, 1 + |t|) < 1$,
    there is a sampling algorithm that computes $\Tr(\rho_0\, e^{iHt } A e^{-iHt })$ to within additive error $\epsilon$ with probability at least $1-\delta$ in time $\widetilde{\OO}\left(\epsilon^{-2}M^3 e^{2d_0} N^3 \right)$, where $d_0 = \deg(A)$ and $\widetilde{\OO}(\cdot)$ hides logarithmic factors in $M, 1/\epsilon$ and $1/\delta$.
\end{corollary}

\medskip
\noindent \textbf{Comparison with previous works.}
Besides Majorana Propagation, which we discussed extensively above, there are many other numerical algorithms for simulating real-time evolution for interacting fermionic systems. Exact diagonalization is a viable approach for small quantum systems, but cannot be applied to large systems due to its runtime scaling. For the weakly interacting fermions setting, methods based on time-dependent variational principle (TDVP) and a fermionic Gaussian ansatz, are expected to have good performance \cite{ShiDemlerCirac2018variational,Mclachlan1964time}. However, due to the limitation of the ansatz, the accuracy achieved by these methods is not systematically improvable. Other ans{\" a}tzes such as matrix-product states and fermionic tensor network states offer more flexibility, but they are typically limited in the representable geometry or do not have provable guarantee for efficient contraction \cite{Cazalilla2002time,WhiteFeiguin2004real,DaiWuEtAlZaletel2025fermionic}.
Monte Carlo methods have been extensively used in many tasks for quantum systems, which include real-time dynamics simulation. Quantum Monte Carlo (QMC) algorithms based on perturbative expansion have been proposed for interacting fermionic systems \cite{werner2010weak,bertrand2019quantum,moutenet2019cancellation,church2021real,cohen2015taming,antipov2017currents}. These methods have achieved impressive empirical success, and some of them have successfully mitigated the dynamical sign problem.
However, the most successful long-time perturbative methods have primarily been demonstrated for quantum impurity problems, where the interacting region is small, while in application to more general interacting systems, a dynamical sign problem can still appear, whose severity can grow exponentially with time, interaction strength, and system size.

In spirit, the sampling algorithm of Theorem~\ref{thm:sample_complexity_l1} combines the continuous-time interaction-expansion viewpoint of continuous-time QMC \cite{rubtsov2005continuous,gull2011continuous,werner2010weak} with the real-time perturbative sampling strategy of diagrammatic QMC \cite{bertrand2019quantum,moutenet2019cancellation}. However, instead of stochastically summing Green's-function diagrams or determinant weights, our algorithm samples Majorana monomials in a Heisenberg-picture Dyson expansion. A distinguishing feature of our approach is that it comes with a rigorous variance and runtime analysis. In particular, by tracking operator growth directly in the Heisenberg picture, we obtain an explicit parameter regime for evolution time $t$ and interaction strength $\lambda$, specified in Theorem~\ref{thm:sample_complexity_l1}, in which the number of samples required to estimate a local observable does not grow with the system size $N$. Thus, within this regime, the sampling part of the algorithm avoids the exponential-in-$N$ variance growth associated with the dynamical sign problem.

The runtime reduction due to Anderson localization of $H^0$ in Corollary~\ref{cor:samplecomplexity_anderson} is conceptually related to a line of work that exploits localization to improve the complexity of simulating quantum dynamics \cite{KimChandranAbanin2014local,huang2015efficient,ehrenberg2022simulation,DeTomasi2019efficiently,chapman2018classical}. In one dimension, many-body localization can lead to logarithmic light cones and consequently to efficient local simulation of real-time dynamics, as shown in \cite{KimChandranAbanin2014local,huang2015efficient}. The complexity-theoretic framework of \cite{ehrenberg2022simulation} can in principle apply more broadly, but its long-time simulation algorithm assumes access to the l-bit description of the Hamiltonian, including a quasi-local unitary that diagonalizes it. The method of \cite{DeTomasi2019efficiently} was also applied to two-dimensional systems, but it proceeds through an approximate effective l-bit Hamiltonian constructed from localized single-particle orbitals, and its asymptotic complexity and error scaling are not established as a rigorous complexity theorem. In comparison, Corollary~\ref{cor:samplecomplexity_anderson} gives a rigorous runtime improvement in arbitrary spatial dimension, within the weak-interaction regime, which is also assumed for two-dimensional simulations in \cite{DeTomasi2019efficiently}.

\medskip
\noindent \textbf{Organization.}
In Section~\ref{sec:prelim} we define the notation we use.
In Section~\ref{sec:convergence} we show convergence theorems in terms of (normalized) Frobenius norm (Section~\ref{sec:convergence_in_F_norm}) and Majorana $1$-norm (Section~\ref{sec:convergence_in_one_norm}). In Section~\ref{sec:algorithm} we present the simulation algorithm based on sampling. In Section~\ref{sec:Anderson} we discuss Anderson localization and the application of our algorithm to it.

\section{Preliminaries}\label{sec:prelim}

Consider $N$ fermionic modes where each mode $i$ is associated with an annihilation operator $\hat{\psi}_i$ and a creation operator $\hat{\psi}^\dag _i$.
These creation and annihilation operators obey the canonical anti-commutation relations, thus generating a CAR algebra. Another convenient way to represent operators acting on the Fock space is through the Majorana operators, which are defined by
\begin{equation}
    \begin{gathered}
      \gamma_{2i-1} = \hat{\psi}^\dag_i + \hat{\psi}_i, \quad
      \gamma_{2i} = \ii (\hat{\psi}^\dag_i - \hat{\psi}_i).
    \end{gathered}
\end{equation}
We denote by $N$ the number of modes.
Let $\Gamma = [N]$ index all the fermionic modes.
The Majorana operators have the following basic properties:
\begin{enumerate}
  \item Hermiticity. $\gamma_i = \gamma_i ^\dag$.
  \item Operator norm one. $\| \gamma_i \| = 1$.
  \item Anti-commutation. $\{ \gamma_i, \gamma_j \} = 2 \delta_{ij}.$ In particular, $\gamma_i^2=1$, $\forall i$.
\end{enumerate}
Define Majorana monomials for a subset $S \subseteq [2N]$ as
\begin{equation}\label{eq:Majorana_monomial}
    \gamma_S := \ii^{|S|(|S|-1)/2 \!\!\!\!\mod 2} \, \prod_{i\in S} \gamma_i,
\end{equation} where the product is ordered increasingly by $i$, and the prefactor included to ensure Hermiticity. We use $\frakI (\cdot) $ to denote the function that returns the prefactor of a Majorana monomial. We use $\majsupp (\gamma_S) = S $ to denote the support in terms of the Majorana index. For $A =\sum_{S\subseteq [2N]} a_S \gamma_S$, we define $\supp(A) = \cup_{|a_S| \neq 0} S$.
Denote by $M_\Gamma$ all the Majorana monomials on $\Gamma$. The commutator of two Majorana monomials can be computed in the following way:

\begin{proposition}\label{prop:monomial_commutator}
    Let $|T|$ be even. When $|S\cap T|$ is odd,
    \begin{equation}
        [\gamma_S, \gamma_T ] = 2 \zeta_{S,T} \, \gamma_{S\Delta T},
    \end{equation} where $S \Delta T = S\cup T - S\cap T$, $\zeta_{S,T}$ is a complex number of modulus $1$ whose value also depends on the order of $S$ and $T$. $[\gamma_S, \gamma_T ] =0$ when $|S\cap T|$ is even.
\end{proposition}

The operator algebra $\calA_\Gamma$ acting on the Fock space $ \mathcal{F}_\Gamma = \operatorname{span}\{ \ket{ n_1, n_2, \cdots, n_{N} } \}_{n_i = 0, 1}$ is generated by the annihilation and creation operators $\{ \hat{\psi}_i, \hat{\psi}^\dag_i\} _{i\in [N]}$ and also by the Majorana operators $\{\gamma_i\}_{i\in [2N]}$. For $S=\{x_1,x_2,\cdots,x_{|S|}\} \subseteq \Gamma$, denote by
$\calF_S =\operatorname{span}\{ \ket{ n_{x_1}, n_{x_2}, \cdots, n_{x_{|S|}} } \}_{n_{x_k} = 0, 1}$
the Fock space on the subset $S$ of modes, and denote by $\calA_S$ the corresponding operator algebra on $\calF_S$. Write $\calA_S^H$ for the subspace of Hermitian operators in $\calA_S$. Any operator $A \in \calA_\Gamma$ has a unique expansion in the Majorana monomial basis,
\begin{equation}
    A = \sum_{S\subseteq [2N]} a_S \: \gamma_S.
\end{equation}
When $A \in \mathcal{A}_\Gamma^H$, the coefficients are all real.
If we index all the Majorana monomials in $M_\Gamma$ in any particular order, we can then represent an operator $A$ as a vector:
\begin{equation} \label{eq:vector_rep}
    \vecrm(A) := (a_S)_{S\subseteq [2N]}.
\end{equation}
The coefficients can be extracted by taking inner product:
\begin{definition}
    We denote the inner product $\braket{A, B} = \overline{\Tr}(A^\dag B)$, where $\overline{\Tr}(\cdot) = \frac{1}{\Tr(I)} \Tr(\cdot)$ is the normalized trace.
\end{definition}
For $S', S \subseteq [2N]$,
we have $\braket{\gamma_{S'}, \gamma_S} = \delta_{S'S}$. The coefficients in \eqref{eq:vector_rep} can therefore be computed as $a_S = \braket{\gamma_S, A}$.

Any superoperator $\mathcal{E}: \mathcal{A}_\Gamma \to \mathcal{A}_\Gamma$ has an induced matrix representation $E$ defined by:
\begin{equation}\label{eq:matriX_rep}
    \vecrm(\mathcal{E}(A)) = E \vecrm (A).
\end{equation}
We write $\gamma \in A$ to mean that the coefficient of $\gamma$ in $A$ is non-zero.
The \textit{degree} of a single Majorana monomial $\gamma_S$ is $|S|$. The degree of an observable $A$ is defined as
\begin{equation} \label{eq:deg_majorana}
    \deg (A) = \max_{\gamma \in A} \deg(\gamma).
\end{equation}
We list the norms we use throughout the paper.
\begin{definition} \label{def:norm_notation}
    Let $\mathcal{V}=\mathbb{C}^n$. We use the following norms for a vector $v = (v_1, v_2, \cdots,$ $ v_n)^\top \in \mathcal{V}$.

    \begin{enumerate}
        \item Euclidean norm: $\|v\| = \sqrt{\sum_{i} |v_i|^2 }$.
        \item $\ell_1$-norm: $\|v\|_{\ell_1} = \sum_{i} |v_i|$.
        \item $\ell_\infty$-norm: $\|v\|_{\ell_\infty} = \max_{i} |v_i|$.
    \end{enumerate}
    We use the following norms for a linear operator $X$ on the vector space $\mathcal{V}$.
    \begin{enumerate}
        \item Operator norm: $\|X\| = \sup_{v \in \mathcal{V}, \|v\| \le 1 } \| X v \| $.
        \item (Normalized) Frobenius norm $\|X\|_F = \sqrt{\overline{\Tr} (X^\dag X) }$, where $\overline{\Tr} = \frac{1}{\Tr(I)} \Tr$ is the normalized trace.
        \item Induced $1$-$1$ norm: $\|X\|_{\onetoone} := \sup_{v \in \mathcal{V}, \|v\|_{\ell_1} \le 1} \|X v \|_{\ell_1} = \max_{j} \sum_{i} |X_{ij}|$.
        \item Induced $\infty$-$\infty$ norm: $\|X\|_{\inftoinf} := \sup_{v \in \mathcal{V}, \|v\|_{\ell_\infty} \le 1} \|X v \|_{\ell_\infty} = \max_{i} \sum_{j} |X_{ij}|$.
    \end{enumerate}
    We use the following norms for an operator $A$ in the operator algebra $\mathcal{A}_\Gamma$.
    \begin{enumerate}
        \item Operator norm: $\|A\| = \sup_{\psi \in \mathcal{F}_\Gamma, \|\psi\| \le 1} \| A \psi \|$, where $\|\psi\|$ denotes the Hilbert space norm of the state $\psi \in \mathcal{F}_\Gamma$.
        \item (Normalized) Frobenius norm: $\|A\|_F = \sqrt{\overline{\Tr} (A^\dag A) }$, where $\overline{\Tr}(\cdot) = \frac{1}{\Tr(I)} \Tr(\cdot)$ is the normalized trace.
        \item Majorana $1$-norm: $\|A\|_\majone := \|\vecrm(A)\|_{\ell_1}$, where $\vecrm(\cdot)$ is the vectorization in the Majorana monomial basis defined in \eqref{eq:vector_rep} and $\| \cdot \|_{\ell_1}$ is the $\ell_1$ norm for vectors defined above.
    \end{enumerate}
\end{definition}

\section{Convergence of perturbative expansion with locality}\label{sec:convergence}

In this section we will derive bounds for the Frobenius norm and the Majorana $1$-norm of the integrand
\begin{equation}\label{eq:integrand_kth_order}
\Big[ \tau^0_{-t_k} (V), \big[\tau^0_{-t_{k-1}}(V), \cdots, [\tau^0_{-t_1}(V), A]\big] \Big]
\end{equation}
in $\widetilde{A}^{(k)}(t)$ \eqref{eq:kth_order}.

The general methodology in this paper is to view each $[\tau^0_{-t_j}(V), \cdot]$ as a linear map on the vector space of Majorana polynomials. Choosing the Majorana monomials as the basis, we represent any operator by the vector of its Majorana-monomial coefficients, as defined in \eqref{eq:vector_rep}. Accordingly, in Section \ref{sec:prelim} we defined the Euclidean norm and the $\ell_1$ norm of vectors, which correspond to the Frobenius norm and the Majorana $1$-norm of operators, respectively. The growth of the norms of the observable $A$ can thus be tracked by studying the corresponding norms of the linear map $[\tau^0_{-t_j}(V), \cdot]$, i.e., its operator norm and its induced $\onetoone$ norm.

We also avoid dealing directly with the quasi-local $\tau^0_{-t}(V)$ by using the identity:
\begin{equation}
\label{eq:three_step_commute}
    [ \tau^0_{-t} (V), \,\cdot\,] = \tau^0_{-t} \left( [ V, \tau^0_{t}( \,\cdot\,) ] \right).
\end{equation}
This decomposition can be interpreted as operator spreading through three processes: time evolution $\tau^0_{t}$, the commutator $[V, \cdot]$, and then time evolution $\tau^0_{-t}$.
Let $T_{t}$, $T_{-t}$, and $Q$ be the matrix representations
of $\tau^0_{t}$, $\tau^0_{-t}$, and $[V, \,\cdot\,]$ as operators on $\calA_\Gamma$, respectively.
More precisely, indexed by subsets of $[2N]$, 
\begin{subequations}
    \begin{equation} \label{eq:Tt_def}
            (T_t)_{S'S} = \langle \gamma_{S'} , \tau^0_t(\gamma_S) \rangle, \quad \text{for } t \in \mathbb{R}
    \end{equation}
    and
    \begin{equation}\label{eq:mat_rep_Q}
        (Q)_{S'S} = \langle \gamma_{S'} , [V, \gamma_S] \rangle.
    \end{equation}
\end{subequations}
$T_{t}$ is unitary.
$Q$ is antisymmetric, i.e., $(Q)_{S'S} = - (Q)_{SS'}$, since $\langle \gamma_{S'} , [V, \gamma_S] \rangle = - \langle \gamma_{S} , [V, \gamma_{S'}] \rangle$.
In effect, we have written the induced matrix representation of $[ \tau^0_{-t} (V), \,\cdot\,]$ as $T_{-t} Q T_t$.
We have
\begin{subequations}\label{eq:vec_comm}
\begin{equation}
    \vecrm \left( \left[\tau^0_{-t} (V), \vecrm^{-1}(\,\cdot\,) \right] \right) = T_{-t} Q T_t,
\end{equation}
\begin{equation} \label{eq:vec_comm_k}
    \vecrm\left(\Big[ \tau^0_{-t_k} (V), \big[\tau^0_{-t_{k-1}}(V), \cdots, [\tau^0_{-t_1}(V), A]\big] \Big]\right) = T_{-t_k} Q T_{t_k} \cdots T_{-t_1} Q T_{t_{1}} \vecrm(A).
\end{equation}
\end{subequations}

For the Frobenius norm, since the free-evolution matrices $T_{\pm t}$ are unitary, their contribution to the norm growth is trivial, i.e. $\|T_{\pm t}\| = 1$. The main task is therefore to control $Q$ on low-degree Majorana polynomials.
By contrast, for the Majorana 1-norm, $\|T_{\pm t}\|_\onetoone$ can grow with time and captures the spreading of Majorana operators under the free evolution. Controlling this growth requires bounded free-Hamiltonian coefficients together with geometric locality on a finite-dimensional lattice. We treat these two cases in Sections \ref{sec:convergence_in_F_norm} and \ref{sec:convergence_in_one_norm}, respectively.

\subsection{Convergence in normalized Frobenius norm with bounded-degree locality}
\label{sec:convergence_in_F_norm}

In this section, we prove the convergence of \eqref{eq:series_expansion} for bounded-degree Hamiltonians defined as follows and show a quasi-polynomial-time algorithm to evaluate it.

\begin{definition}[Bounded-degree locality] \label{def:bounded_deg_local}
    A Hamiltonian $H = \sum_{k} H_k$ is $(\frakd, \kappa)$-bounded-degree local if, for any single Majorana mode, the number of terms in $H$ involving it is at most a constant $\frakd$, and $\deg(H_k) \le \kappa$ for some constant $\kappa$ for all $k$.
\end{definition}

Denote by $\calP_{p}$ the projection onto the subspace of the vector representation of Majorana operators with degree at most $p$.  We will upper bound several matrix norms of $Q\calP_p$ in this section, where $Q$ is the induced matrix representation of $[V, \cdot]$ as operators on $\calA_\Gamma$.

To capture the growth of degree, let
\begin{equation} \label{eq:vec_comm_A_j}
    A_j = \big[\tau^0_{-t_{j} } (V), \cdots [ \tau^0_{-t_1} (V) , A  ]\big]
\end{equation}
Note that $\tau^0_{-t}(\cdot)$ does not change the degree, while $[V, \cdot ]$ increases the degree by at most $2$. Therefore $\deg(A_j) \le 2j + \deg(A) = 2j + d_0$.
\eqref{eq:vec_comm_A_j} can be alternatively written as:
\begin{equation}\label{eq:restricted_transition_series}
    \vecrm(A_j) = T_{-t_j} \big(Q \calP_{2j-2+d_0} \big)T_{t_j}\cdots T_{-t_1} \big(Q \calP_{d_0}\big) T_{t_1} \vecrm(A).
\end{equation}
We now upper bound the induced $\onetoone$, the induced $\inftoinf$, and the operator norm of $Q\calP_p$. The proof of the convergence (Theorem \ref{thm:Fnorm_convergence_intro}) then follows.

\begin{lemma} \label{lem:commutator_matrix_norm}
Consider a quartic interaction Hamiltonian $V$ that is $(\frakd, 4)$-bounded-degree local (Definition~\ref{def:bounded_deg_local}), with all coefficients bounded by $1$ in absolute value. Let $Q$ be the matrix representation of $[V, \cdot]$, as defined in \eqref{eq:mat_rep_Q}.
    Then, for $\|\cdot\|_\beta$ being the induced $1$-$1$ norm, the induced $\infty$-$\infty$ norm, or the operator norm:
    \begin{equation*}
        \|Q \calP_p\|_\beta \le 2p \frakd.
    \end{equation*}
\end{lemma}

\begin{proof}
Let $\Omega$ denote the set of all Majorana modes.
For $S, S' \subseteq \Omega$, $(Q)_{S'S} = \braket{\gamma_{S'}, [V, \gamma_S]}$, where $\braket{X, Y} = \frac{1}{\Tr(I)} \Tr(X^\dag Y)$. When $p=0$, since $[V, I] = 0$, $Q \mathcal{P}_0 = 0$, the bound holds. Now we focus on $p > 0$.

We will first bound $\|Q \mathcal{P}_p \|_{\onetoone}$.
We fix a $S\subseteq \Omega$ where $|S| \le p$ and consider $\sum_{S' \subseteq \Omega} |(Q \calP_p)_{S'S} | $. Let $V = \sum_{i,j,k,l\in\Omega} v_{ijkl} \gamma_{i}\gamma_j\gamma_k\gamma_l$. Here and below, the sum is over pairwise distinct indices $i,j,k,l$ since $V$ is quartic. First, consider $|S|=p$. 
\begin{equation*}
    \begin{aligned}
        & \sum_{S' \subseteq \Omega } |(Q\calP_p)_{S'S} |
        = \sum_{S' \subseteq \Omega } |(Q)_{S'S} |
        = \sum_{S' \subseteq \Omega} | \braket{ \gamma_{S'} , [V, \gamma_S] } | \\
        = & \sum_{S' \subseteq \Omega} \sum_{i,j,k,l\in \Omega} |v_{ijkl}| \cdot | \braket{\gamma_{S'}, [\gamma_i\gamma_j\gamma_k\gamma_l, \gamma_S]} | \\
        \le & \sum_{S' \subseteq \Omega} \sum_{i,j,k,l\in \Omega} 2 \cdot \istrue \{
            S' = \{i,j,k,l\} \Delta S \text{ AND } S \cap \{i,j,k,l\} \neq \emptyset \text{ AND } v_{ijkl}\neq 0
        \} \\
        = & \sum_{i,j,k,l\in \Omega} 2 \cdot \istrue \{
            S \cap \{i,j,k,l\} \neq \emptyset \text{ AND } v_{ijkl}\neq 0
        \} \\
        \le & 2p\frakd,
    \end{aligned}
\end{equation*}
The last inequality is given by bounded-degree locality: To count the number of ordered tuples $(i,j,k,l)$ such that $\{i,j,k,l\}\cap S \neq \emptyset$ and $v_{ijkl}\neq 0$, since at least one mode in $S$ needs to be in the $S\cap \{i,j,k,l\}$, we first pick that mode from $S$, of which there are $p$ choices. Then we count how many ordered tuples $(i,j,k,l)$ such that $\{i,j,k,l\}$ intersects $S$ at that mode, of which there are at most $\frakd$ choices.
Therefore, $\sum_{i,j,k,l\in \Omega} \istrue \{S \cap \{i,j,k,l\} \neq \emptyset \text{ AND } v_{ijkl}\neq 0\} \le p\frakd$, and $\sum_{S' \subseteq \Omega } |(Q)_{S'S} | \le 2p\frakd$.

When $|S| < p$, similarly  $\sum_{S' \subseteq \Omega } |(Q)_{S'S} | \le 2 |S| \frakd < 2p \frakd$. Since $\|Q \mathcal{P}_p \|_{\onetoone}$ is the maximum absolute column sum, we have $\|Q \mathcal{P}_p \|_{\onetoone} \le 2p\frakd$.

We next consider $\|Q \mathcal{P}_p \|_{\inftoinf}$, which is the maximum absolute row sum.
We fix a $S'\subseteq \Omega$ where $|S'| \le p+2$, because $Q \mathcal{P}_p$ can only increase or decrease the degree by $2$. First, suppose that $|S'| \le p$. Since $Q$ is skew-symmetric,
\[
    \sum_{S \subseteq \Omega} | (Q\calP_p)_{S'S}| \le
    \sum_{S \subseteq \Omega} | (Q)_{S'S}| = \sum_{S \subseteq \Omega} | (Q)_{SS'}| \le 2p\frakd.
\]

Then, suppose that $p < |S'| \le p+2$. Since $|S| \le p $ and $|S'| > p$, $Q$ must increase the degree by $2$. Therefore,
\[
    \sum_{S \subseteq \Omega} | (Q\calP_p)_{S'S}| =
    \sum_{S \subseteq \Omega: |S| = |S'|-2} | (Q)_{S'S}|.
\]
Then,
\begin{equation*}
    \begin{aligned}
        & \sum_{S \subseteq \Omega: |S| = |S'|-2} | (Q)_{S'S} | = \sum_{S \subseteq \Omega: |S| = |S'|-2} | \braket{\gamma_{S'}, [V, \gamma_S]} | \\
        = & \sum_{S \subseteq \Omega: |S| = |S'|-2} \: \sum_{i,j,k,l \in \Omega} |v_{ijkl}| \cdot | \braket{\gamma_{S'}, [\gamma_i\gamma_j\gamma_k\gamma_l, \gamma_S]} | \\
        \le & \sum_{S \subseteq \Omega: |S| = |S'|-2} \: \sum_{i,j,k,l\in \Omega} 2 \cdot \istrue \{
            S' = \{i,j,k,l\} \Delta S \text{ AND } |S' \cap \{i,j,k,l\}|= 3 \text{ AND } v_{ijkl}\neq 0
        \} \\
        = & \sum_{i,j,k,l \in \Omega} 2 \cdot\istrue\{ |S' \cap \{i,j,k,l\}|=3 \text{ AND } v_{ijkl}\neq 0\}
    \end{aligned}
\end{equation*}

We show below that $\sum_{i,j,k,l\in \Omega: v_{ijkl}\neq 0} \istrue\{ |S' \cap \{i,j,k,l\}| = k \} \le \frac{1}{k}\frakd |S'|$.

Consider an undirected bipartite graph $(L,E)$ where $L = M \cup N$ is the set of all vertices and $E$ is the set of all edges. $M, N$ are two disjoint subsets of vertices. Let $|M| = |S'|$ and each vertex in $M$ represents a Majorana mode in $S'$.
Let $N$ represent all the terms in the interaction Hamiltonian $V$, where each vertex represents a term. When a term in $V$ involves a Majorana mode in $S'$, we connect an edge from the vertex in $N$ corresponding to that term, to the vertex in $M$ representing that Majorana mode. This forms an interaction graph. Let $N_k = \{ w \in N: \deg(w) = k \}$, and we have $$\sum_{i,j,k,l \in \Omega: v_{ijkl}\neq 0} \istrue\{ |S' \cap \{i,j,k,l\}| = k \} = |N_k|.$$ Since every vertex in $N_k$ has degree $k$, we have $|N_k|k\leq |E|$.
On the other hand, because the degree of every vertex in $M$ is at most $\frakd$, the number of edges going out from $M$ is at most $|S'|\frakd$.
Therefore $|E| \le |S'|\frakd$. Combining the two inequalities involving $|E|$, we obtain $|N_k| \le \frac{1}{k} |S'| \frakd$.

Therefore when $|S'| \le p$, $\sum_{S \subseteq \Omega} | (Q\calP_p)_{S'S}| \le 2p\frakd$, and when $ p < |S'| \le p+2$, $\sum_{S \subseteq \Omega} | (Q\calP_p)_{S'S}| \le 2/3 |S'| \frakd \le \frac{2(p+2)}{3} \frakd$. When $p\ge 1$, we have $\frac{2(p+2)}{3} \le 2p$. Therefore, taking the maximum, $\|Q\calP_p\|_\inftoinf \le 2p\frakd$.

Lastly, for the operator norm:
\[
\|Q \mathcal{P}_p\| \le \sqrt{\|Q \mathcal{P}_p\|_\onetoone \|Q \mathcal{P}_p\|_\inftoinf}   \le 2p\frakd.
\]

\end{proof}

\begin{proof}[Proof of Theorem \ref{thm:Fnorm_convergence_intro}]
    For the integrand in $\widetilde{A}^{(k)}(t)$ defined in \eqref{eq:kth_order}, we have
    \begin{equation}
    \begin{aligned}
        & \Big\|\Big[ \tau^0_{-t_k} (V), \big[\tau^0_{-t_{k-1}}(V), \cdots, [\tau^0_{-t_1}(V), A]\big] \Big] \Big\|_F \\
        & =  \| T_{-t_k} Q \mathcal{P}_{2k-2+d_0} T_{t_k} \cdots T_{-t_1} Q \mathcal{P}_{d_0} T_{t_1} \vecrm(A) \| \\
        & \le  \left( \prod_{j=0}^{k-1} \|Q \mathcal{P}_{2j + d_0}\| \right) \|\vecrm(A)\|
        \le \left( \prod_{j=0}^{k-1} \left( 2 (2j+d_0) \frakd \right) \right) \|A\|_F \\
        & \le 2^k (2k-2 + d_0)^k \frakd^k \|A\|_F \le 2^k\frakd^k k! e^{2k+d_0-2} \|A\|_F \le (2 e^2 \frakd)^k k! e^{d_0-2} \|A\|_F,
    \end{aligned}
\end{equation}
where we used that $x^k\le k!e^x$ for $x\ge 0$ and $\|T_{\pm t}\|=1$ for all $t$.
Therefore,
\begin{equation}\label{eq:bound_aikf}
    \begin{aligned}
        \| \widetilde{A}^{(k)}(t) \|_F & =  \left\| (i\lambda)^k \int_0^{t} \dd t_k \int_0^{t_{k}} \dd t_{k-1} \cdots \int_0^{t_{2}} \dd t_1 \, \Big[ \tau^0_{-t_k} (V), \big[\tau^0_{-t_{k-1}}(V), \cdots, [\tau^0_{-t_1}(V), A]\big] \Big]  \right\|_F\\
        & \le \frac{\lambda^k|t|^k}{k!} (2 e^2 \frakd)^k e^{d_0-2} k!  \|A\|_F
        \le (2 e^2\lambda |t| \frakd)^k e^{d_0-2} \|A\|_F .
    \end{aligned}
\end{equation}
\end{proof}

To evaluate the time-ordered integral \eqref{eq:kth_order} of Majorana polynomials over $(t_1, t_2, \cdots, t_k)$ using Monte Carlo methods, we repeat the following procedure for a sufficient number of times: uniformly sample  $(t_1, t_2, \cdots, t_k)$ from the time-ordered simplex $\left\{(t_1,\ldots,t_k): 0\leq t_1\leq t_2\leq\cdots\leq t_k\leq t\right\}$, evaluate the integrand as a Majorana polynomial, multiply by the weight $(i\lambda)^kt^k/k!$, and retain the result as one sample.
We then average the samples.
At fixed times $(t_1, t_2, \cdots, t_k)$,
the integrand can be evaluated as a polynomial of Majorana operators.
Since there are $\Or(N^{k})$ Majorana monomials of degree $\Or(k)$, an $N^{\Or(k)}$-time algorithm is readily available through linear algebra operations within the subspace spanned by Majorana monomials up to degree $\Or(k)$.

\begin{proof}[Proof of Corollary \ref{cor:quasi-polynomial-time}]
    When $|t| < 0.99/(2e^2 \lambda \frakd)$, by Theorem \ref{thm:Fnorm_convergence_intro}, we truncate series \eqref{eq:series_expansion} at order $K = \Or(\log(e^{d_0 } \|A\|_F/\epsilon))$, controlling the truncation error to $\epsilon$.
    We then use the Monte Carlo method to evaluate each order.

    Now we fix an order $k\in [K]$.  We sample $k$ independent random variables
uniformly from $[0,t]$ and sort them to obtain
\begin{equation}
    0\le t_1\le t_2\le \cdots \le t_k\le t\ .
\end{equation}
 The resulting tuple $(t_1,\ldots,t_k)$ is uniformly distributed over the
time-ordered simplex, whose volume is $t^k/k!$. For each such tuple, define the Majorana polynomial
\begin{equation}
    S=(i\lambda)^k\frac{t^k}{k!}
\Big[ \tau^0_{-t_k}(V),
\big[\tau^0_{-t_{k-1}}(V),\ldots,
[\tau^0_{-t_1}(V),A]\big]\Big]\ .
\end{equation}
We repeat this procedure independently $m$ times and denote the resulting
samples by $S_1,S_2,\ldots,S_m$. Each $S_j$ is a Majorana polynomial of
degree $\Or(k)$.

By the bound used in the proof of Theorem
\ref{thm:Fnorm_convergence_intro},
\begin{equation}
    \|\widetilde A^{(k)}(t)\|_F\le
\alpha^k C\|A\|_F,
\qquad
\|S_j\|_F\le
\alpha^k C\|A\|_F\ ,
\end{equation}
for all $j\in[m]$, where
$\alpha=2e^2\lambda\frakd|t|<0.99$ and
$C=e^{d_0-2}$.
Since the samples are independent and unbiased,
\begin{equation}
    \E S_j=\widetilde A^{(k)}(t)\ .
\end{equation}

    Furthermore, 
    \begin{equation}
        \E \left\| \frac{1}{m } \sum_{i=1}^m S_i - \tilde{A}^{(k)}(t)  \right\|_F^2 =  \frac{1}{m} \E \left\| S_i - \tilde{A}^{(k)}(t) \right\|_F^2 \le \frac{1}{m} \left( \alpha^k \|A\|_FC \right)^2.
    \end{equation}
    The random variable $Z:= \left\| \frac{1}{m } \sum_{i=1}^m S_i - \tilde{A}^{(k)}(t)  \right\|_F$  quantifies the error. The above bound implies that 
    \begin{equation}
    \E Z    \le \sqrt{\E Z^2}\le \frac{1}{\sqrt{m}}\alpha^k \|A\|_FC\ .
     \end{equation}
    
    Since changing one $S_i$ causes $Z$ to change by at most
$\frac{2}{m}\alpha^k\|A\|_F C$, McDiarmid's inequality gives
\begin{equation}
    \PP(Z-\E Z\ge u)
    \le
    \exp\left(
    -\frac{m u^2}
    {2(\alpha^k\|A\|_F C)^2}
    \right).
\end{equation}
    
    Therefore, by McDiarmid's inequality,
\begin{equation}
    \PP\left(
        Z > \E Z + \frac{\epsilon}{K}
    \right)
    \le
    \exp\left(
        -\frac{m\epsilon^2}
        {2K^2\left(\alpha^k\|A\|_F C\right)^2}
    \right).
\end{equation}
Since there are $K$ orders, we require the concentration contribution for each order to be at most $\epsilon/K$ with probability at least $1-\delta/K$. Hence, to make this probability at most $\delta/K$, it suffices to choose
\begin{equation} \label{eq:sample_number}
    m =
    \Or\left(
        \frac{K^2\left(\alpha^k\|A\|_F C\right)^2}{\epsilon^2}
        \log\frac{K}{\delta}
    \right).
\end{equation}
We also require $\E Z < \epsilon/K$. Since
\begin{equation}
    \E Z
    \le
    \frac{1}{\sqrt{m}}\alpha^k\|A\|_F C,
\end{equation}
it suffices to have
\begin{equation}
    m >
    \frac{K^2\left(\alpha^k\|A\|_F C\right)^2}{\epsilon^2},
\end{equation}
which is already satisfied by the scaling in
\eqref{eq:sample_number}.
    Now our total error is $3\epsilon$, one $\epsilon$ due to the truncation of the series \eqref{eq:series_expansion}, one $\epsilon$ due to the concentration error per each order, and one $\epsilon
    $ due to the expectation error $\E Z$ per each order. The total runtime is
    \begin{equation}
        \Or\left( \frac{K^2 \left(\|A\|_FC\right)^2}{\epsilon^2} \log \frac{K}{\delta}  \cdot N^{\Or(K)} \right).
    \end{equation}
    Absorbing the $\frac{1}{\epsilon^2}$ factor into $N^{\Or(K)}$ gives the runtime in the corollary.

\end{proof}

\subsection{Convergence in Majorana $1$-norm with geometric locality}
\label{sec:convergence_in_one_norm}

The previous section establishes exponential convergence of the series \eqref{eq:series_expansion} in the Frobenius norm under bounded-degree locality condition and gives a quasi-polynomial-time algorithm. In this section, we establish exponential convergence of \eqref{eq:series_expansion} in the Majorana $1$-norm, which upper bounds the operator norm, under geometric locality condition. More precisely, the fermionic modes are arranged on a $D$-dimensional lattice where $D$ is a constant, and they interact locally: there is an interaction range $r_0$ beyond which no two modes  interact. In particular, we consider a distance metric $d(\cdot,\cdot)$ to represent the distance between pairs of Majorana modes. If a pair of Majorana modes have distance zero, we say they are on the same site of the lattice.

\begin{definition}[Geometric locality] \label{def:geo_local}
    Given a metric $d(i,j)$ for Majorana modes $i, j\in [2N]$. A Hamiltonian $H = \sum_k H_k$ is $(r_0, \frakd)$-geometrically local if there is no Hamiltonian term $H_k$ involving $i$ and $j$ simultaneously for $d(i, j) > r_0$. Furthermore, for any single mode, the number of terms in $H$ involving it is bounded by a constant $\frakd$.
\end{definition}
Note that for lattice Hamiltonians, the fact that the number of terms in $H$ involving any single mode is bounded by $\frakd$ is a consequence of finite interaction range and finite dimension, and the upper bound of $\frakd$ can be determined by the interaction range $r_0$ and the dimension $D$. Still, we adopt a separate $\frakd$ parameter here for convenience.

We consider the induced matrix representation of $[ \tau^0_{-t} (V), \,\cdot\,]$, written as $T_{-t} Q T_t$, as defined in \eqref{eq:matriX_rep}, where $T_{\pm t}$ is the induced matrix representation of $\tau^0_{\pm t} (\cdot)$, and $Q$ is the induced matrix representation of $[V, \cdot]$.  Our main goal is to bound $\|T_{-t}QT_t\calP_p\|_\onetoone$. We first prove a lemma that allows us to restrict attention to degree-$1$ Majorana polynomials, i.e. $\|T_{-t}QT_t\calP_1\|_\onetoone$. Since $T_{-t} Q T_t \calP_1 = T_{-t} \calP_3 Q \calP_1 T_t \calP_1$, we will then bound $\|T_{-t} \calP_3\|_\onetoone, \|Q\calP_1\|_\onetoone$ and $\|T_t\calP_1\|_\onetoone$, where $\|Q\calP_1\|_\onetoone$ already follows from Lemma \ref{lem:commutator_matrix_norm}. Combining these together yields the proof of Theorem \ref{thm:convergence_maj_one}.

\begin{lemma}\label{lem:reduction2deg_one}
    \(
    \|T_{-t} Q T_{t} \calP_p \|_\onetoone
    \le p \|T_{-t} Q T_{t} \calP_1 \|_\onetoone.
    \)
\end{lemma}

\begin{proof}
    By definition
    \begin{equation}
        \|T_{-t} Q T_{t} \calP_p \|_\onetoone = \sup_{\gamma_S: |S| \le p }
        \left\| \left[ \tau^0_{-t}(V), \gamma_S \right] \right\|_\majone.
    \end{equation}
Note that we can restrict the supremum to $|S|>0$ rather than $|S|\ge 0$, because when $|S|=0$, i.e., $\gamma_{\emptyset} = I$, the commutator vanishes.

Consider a Majorana monomial $\nu = \frakI (\nu)\, \nu_{1} \nu_{2} \cdots \nu_{J}$ where $J = \deg(\nu)$. Observe that
\begin{equation} \label{eq:commutator_expansion}
    \begin{aligned}
        [\tau^0_{-t} V, \nu] & = \frakI (\nu) \sum_{j=1}^{J} \Big( \prod_{m=1}^{j-1} \nu_m \Big)  \,  [ \tau^0_{-t} V, \nu_{j} ] \, \Big( \prod_{m'=j+1}^{J} \nu_{m'} \Big) , \\
        & = \frakI (\nu) \sum_{j=1}^{J} \Big( \prod_{m=1}^{j-1} \nu_m \Big)  \,  \tau^0_{-t} [  V, \tau^0_{t}(\nu_{j}) ] \, \Big( \prod_{m'=j+1}^{J} \nu_{m'} \Big).
    \end{aligned}
\end{equation}
Then,
\begin{equation}\label{eq:one_norm_iter}
        \| [\tau^0_{-t} V, \nu] \|_{\majone} \le   \deg(\nu) \cdot \sup _{j\in [2N] } \big\| \big[ \tau^0_{-t} [  V, \tau^0_{t}(\gamma_{j}) ] \big] \big\|_{\majone}
\end{equation}

The second factor on the right-hand side is precisely the definition of $\|T_{-t} Q T_{t} \calP_1 \|_\onetoone$.
We therefore have
\begin{equation}
    \sup_{\gamma_S: |S| \le p }
        \left\| \left[ \tau^0_{-t}(V), \gamma_S \right] \right\|_\majone \le p \sup_{\gamma_S: |S| \le 1 }
        \left\| \left[ \tau^0_{-t}(V), \gamma_S \right] \right\|_\majone,
\end{equation} which proves the Lemma.
\end{proof}

We note that $T_t \calP_1$ relates to a matrix representation of the form $e^{i\htilde t}$:
\begin{subequations} \label{eq:htilde_def_all}

\begin{equation}
    \label{eq:htilde_relation_id}
    \left(T_t \calP_1\right) \vecrm (I) = \vecrm(I),
\end{equation}

\begin{equation} \label{eq:htilde_relation_deg1}
    \begin{aligned}
        \left(T_t \calP_1\right) \vecrm (\gamma_j)
        & = \vecrm \left( e^{iH^0 t } \gamma_j e^{-iH^0 t}  \right) \\
        & =  \sum_i \braket{ \gamma_i, e^{iH^0 t } \gamma_j e^{-iH^0 t}  } \vecrm(\gamma_i) \\
        & = \sum_i \left( e^{i\htilde t} \right)_{ij} \vecrm(\gamma_i),
    \end{aligned}
\end{equation} with
\begin{equation} \label{eq:htilde_def}
    \htilde \in \mathbb{C}^{2N\times 2N}, (\htilde)_{ij} = \braket{ \gamma_i,  [H^0, \gamma_j]} = 2h_{ij} - 2 h_{ji} = 4h_{ij},
\end{equation} where the $h_{ij}$'s are the Hamiltonian coefficients as in \eqref{eq:interacting_fermions_intro}.
\end{subequations}
$\htilde$ is Hermitian, and is $\frakd/2$-sparse in both its rows and columns (the $1/2$ factor arises because two symmetric terms $\gamma_i\gamma_j$ and $\gamma_j \gamma_i$ in $H^0$ give rise to the same $\htilde_{ij}$, even though they are counted as two distinct terms when calculating $\frakd$). Consequently, $\|\htilde\|_\onetoone \le 4 \cdot \frakd/2 = 2 \frakd$. Moreover, $\htilde_{ij} = 0$ if $d(i,j) > r_0$, which helps us control the operator spreading generated by $H^0$.

We now consider the $\onetoone$ norm of $T_t\calP_1$. By \eqref{eq:htilde_relation_id} and \eqref{eq:htilde_relation_deg1}, $\|T_t \calP_1 \|_\onetoone = \max\{ 1, \|e^{i\htilde t} \|_\onetoone \}$.
The $\onetoone$ norm is the maximum $\ell_1$ norm of a column. Because every column of $e^{i\htilde t}$ has Euclidean norm $1$, its $\ell_1$ norm is at least $1$, and hence $\|e^{i\htilde t}\|_{\onetoone}\geq 1$. Therefore
\begin{equation}
    \|T_t \mathcal{P}_1\|_\onetoone = \max\{ 1, \|e^{i\htilde t} \|_\onetoone \} = \|e^{i\htilde t} \|_\onetoone.
\end{equation}

\begin{lemma}\label{lem:eiht_bnd}
    Consider the quadratic Hamiltonian $H^0 = \sum_{i,j=1}^{2N} h_{ij} \gamma_i \gamma_j $, where $h$ is a purely imaginary antisymmetric matrix whose entries are bounded by $1$ in absolute value. Let the Majorana modes be arranged on a $D$-dimensional lattice, and let $H^0$ be $(r_0, \frakd)$-geometrically local. Let $T_t$ be the matrix representation of the time-evolution $\tau^0_t = e^{iH^0 t} (\cdot) e^{-iH^0 t}$ given in \eqref{eq:Tt_def} and $\htilde$ given in \eqref{eq:htilde_def}.
    Then, we have
    \begin{equation*}
        \left\| T_{t} \calP_1 \right\|_\onetoone = \left\| e^{i\htilde t} \right\|_\onetoone \le C_D\lceil 4e \frakd |t| r_0 \rceil^{D/2}  + 1,  \quad \forall t \in \mathbb{R},
    \end{equation*}
    where $C_D>0$ is a constant that depends only on $D$. In particular, $C_D$ is chosen such that the number of modes in a sphere of radius $l>0$ is upper bounded by $C_D^2 \lceil l\rceil^{D}$.

\end{lemma}

\begin{proof}
    Consider two modes $\alpha, \beta \in [2N]$. Let the bra-ket notation denote the corresponding basis for the matrix $\htilde$, e.g., $\bra{\beta} \htilde \ket{\alpha} = (\htilde)_{\beta\alpha}$.
    By Taylor expansion, we have
    \begin{equation}
        \braket{\beta| e^{i\htilde t} |\alpha} =
        \sum_{p=0}^{\infty} \braket{\beta| (i\htilde)^p |\alpha} \frac{t^p}{p!}.
    \end{equation}
    Consequently,
    \begin{equation}
        \sum_{\beta: d(\beta,\alpha) \ge l} |\braket{\beta| e^{i\htilde t} |\alpha}| \le
        \sum_{p=0}^{\infty} \sum_{\beta: d(\beta,\alpha) \ge l} |\braket{\beta| (i\htilde)^p |\alpha}| \frac{|t|^p}{p!}.
    \end{equation}
    Observe first that for $d(\beta,\alpha) \ge l$, $|\braket{\beta| (i\htilde)^p |\alpha}| > 0$ only when $p\ge \lceil l/r_0 \rceil$. Second, since $\|\htilde\|_\onetoone \le 2 \frakd$, we have
    \begin{equation}
        \sum_{\beta: d(\beta,\alpha) \ge l} |\braket{\beta| (i\htilde)^p |\alpha}| \le \sum_{\beta} |\braket{\beta| (i\htilde)^p |\alpha}| \le \|\htilde\|_\onetoone^p \le (2\frakd)^p.
    \end{equation}
    It follows that
    \begin{equation}
        \sum_{\beta: d(\beta,\alpha) \ge l} |\braket{\beta| e^{i\htilde t} |\alpha}| \le
        \sum_{p= \lceil l/r_0 \rceil}^{\infty} (2\frakd)^p \frac{|t|^p}{p!} \le \sum_{p=\lceil l/r_0 \rceil}^\infty \left( \frac{2e\frakd |t|}{p} \right)^p \le 1\ ,
    \end{equation}
    where in the last inequality  we have assumed $l \ge 4e\frakd |t| r_0 $.

    On the other hand, $\sum_{\beta: d(\beta,\alpha) < l} |\braket{\beta| e^{i\htilde t} |\alpha} |^2 \le 1$. Therefore,
    \begin{equation}
        \sum_{\beta: d(\beta,\alpha) < l} |\braket{\beta| e^{i\htilde t} |\alpha} | \le \sqrt{|\{\beta: d(\beta,\alpha) < l\}|} \leq  C_D {\lceil l\rceil ^{D/2}},
    \end{equation}
    where $C_D>0$ is a $D$-dependent constant that satisfies $|\{\beta: d(\beta,\alpha) < l\}|\leq C_D^2 \lceil l \rceil^D$. It exists because the modes are arranged on a $D$-dimensional lattice.
    Combining the above, we have
    \begin{equation}
        \sum_{\beta: \beta \in \Gamma} |\braket{\beta| e^{i\htilde t} |\alpha} | \le C_D \lceil 4e \frakd |t| r_0 \rceil ^{D/2}  + 1.
    \end{equation}
\end{proof}

\begin{lemma} \label{lem:1norm_submulplicativity}
    Let $T_t$ be the same as in Lemma \ref{lem:eiht_bnd}. Then,
    \(
        \left\| T_{t} \calP_p \right\|_\onetoone \le \left\| T_t \calP_1 \right\|_\onetoone^p.
    \)
\end{lemma}

\begin{proof}

    Observe that for any Majorana monomial $\nu = \frakI(\nu)\nu_1 \nu_2 \cdots \nu_J$,
    \begin{equation}
        \tau^0_t(\nu_1 \nu_2 \cdots \nu_J) = \tau^0_t(\nu_1) \tau^0_t(\nu_2) \cdots \tau^0_t(\nu_J).
    \end{equation} The Majorana $1$-norm is submultiplicative, i.e. $\|XY\|_\majone \le \|X\|_\majone \|Y\|_\majone$, since if $X = \sum_i a_i \gamma_i, Y= \sum_j b_j \gamma_j$, we have
    \begin{equation}
        \left\| \sum_i a_i \gamma_i \sum_j b_j \gamma_j \right\|_\majone =  \left\| \sum_{i, j} a_ib_j \gamma_i \gamma_j \right\|_\majone \le \sum_{i,j} |a_i b_j| = \sum_i |a_i| \sum_j |b_j|.
    \end{equation}
    Therefore, we obtain
    \begin{equation}
        \| \tau^0_t(\nu_1 \nu_2 \cdots \nu_J)\|_\majone \le \| \tau^0_t(\nu_1)\|_\majone \| \tau^0_t(\nu_2)\|_\majone \cdots \|\tau^0_t(\nu_J)\|_\majone \le \|T_t \calP_1\|_\onetoone^J .
    \end{equation}
    For $T_t\calP_p$, we consider $J\le p$. Since $\|T_t \calP_1\|_\onetoone \ge 1$, we obtain
    \begin{equation}
        \left\| T_{t} \calP_p \right\|_\onetoone \le \left\| T_t \calP_1 \right\|_\onetoone^p.
    \end{equation}
\end{proof}

\begin{corollary} \label{cor:TQTPp_1to1norm}
    $\|T_{-t} Q T_t \calP_p \| _\onetoone \le p W_t$, where $W_t = 2 \frakd \left( 1 + C_D  \lceil 4e \frakd |t| r_0 \rceil^{D/2}  \right)^4$.
\end{corollary}

\begin{proof}
First, by Lemma \ref{lem:reduction2deg_one},
\begin{equation}
    \|T_{-t} Q T_t \calP_p \| _\onetoone \le p \|T_{-t} Q T_t \calP_1 \| _\onetoone.
\end{equation}
Then
\begin{equation}
        \|T_{-t} Q T_t \calP_1 \| _\onetoone
        = \| T_{-t} \calP_3 Q \calP_1 T_t \calP_1 \|_\onetoone
        \le \| T_{-t} \calP_3 \|_\onetoone \| Q \calP_1 \|_\onetoone \| T_t \calP_1 \|_\onetoone.
\end{equation}
By Lemma \ref{lem:1norm_submulplicativity},
\begin{equation}
    \| T_{-t} \calP_3 \|_\onetoone \le \| T_{-t} \calP_1 \|_\onetoone ^3.
\end{equation}
Since a geometrically local Hamiltonian satisfies the bounded-degree locality condition in Lemma \ref{lem:commutator_matrix_norm}, we have
\begin{equation}
    \|Q\calP_1\|_\onetoone \le 2\frakd.
\end{equation}
Combining this bound with Lemma \ref{lem:eiht_bnd}, we obtain
\begin{equation}
    \|T_{-t} Q T_t \calP_p \| _\onetoone \le p \|T_{-t} Q T_t \calP_1 \| _\onetoone \le p \|T_{-t}\calP_1\|_\onetoone^3 \cdot 2\frakd \cdot \|T_{t}\calP_1\|_\onetoone \le p W_t.
\end{equation}

\end{proof}

\begin{proof}[Proof of Theorem \ref{thm:convergence_maj_one}]

    By \eqref{eq:vec_comm_k},
    \begin{equation}
        \vecrm\left(\Big[ \tau^0_{-t_k} (V), \big[\tau^0_{-t_{k-1}}(V), \cdots, [\tau^0_{-t_1}(V), A]\big] \Big]\right) = T_{-t_k} Q T_{t_k} T_{-{t_{k-1}}} Q T_{{t_{k-1}}} \cdots T_{-t_1} Q T_{t_{1}} \vecrm(A),
    \end{equation}
    which can alternatively be written as
    \begin{equation}
        (T_{-t_k} Q T_{t_k} \calP_{d_0 + 2k - 2} )
        (T_{-{t_{k-1}}} Q T_{{t_{k-1}}}  \calP_{d_0 + 2k - 4}  )
        \cdots (T_{-t_1} Q T_{t_{1}}  \calP_{d_0} ) \vecrm(A),
    \end{equation}
    because each application of $T_{-t_j} Q T_{t_j}$ increases the degree by at most $2$.

    By Corollary \ref{cor:TQTPp_1to1norm}, we have
    \begin{equation}
        \|T_{-t_j} Q T_{t_j} \calP_{d_0 + 2j - 2}\|_\onetoone \le (d_0 + 2j - 2) W_{t_j} \le (d_0 + 2j - 2) W_t.
    \end{equation}
    After $k$ applications, we have
    \begin{equation}
        \left\|\Big[ \tau^0_{-t_k} (V), \big[\tau^0_{-t_{k-1}}(V), \cdots, [\tau^0_{-t_1}(V), A]\big] \Big]  \right\|_\majone \le
        W_t^k \prod_{j=0}^{k-1} \left( d_0 + 2j \right)  \|A\|_\majone
    \end{equation}

    We then have
    \begin{equation}
        \begin{aligned}
            \| \widetilde{A}^{(k)}(t) \|_\majone & =  \left\| (i\lambda)^k \int_0^{t} \dd t_k \int_0^{t_{k}} \dd t_{k-1} \cdots \int_0^{t_{2}} \dd t_1 \, \Big[ \tau^0_{-t_k} (V), \big[\tau^0_{-t_{k-1}}(V), \cdots, [\tau^0_{-t_1}(V), A]\big] \Big]  \right\|_\majone\\
            & \le \lambda^k \int_0^{t} \dd t_k \int_0^{t_{k}} \dd t_{k-1} \cdots \int_0^{t_{2}} \dd t_1 \, \left\|\Big[ \tau^0_{-t_k} (V), \big[\tau^0_{-t_{k-1}}(V), \cdots, [\tau^0_{-t_1}(V), A]\big] \Big]  \right\|_\majone\\
            & \le \lambda^k \frac{|t|^k}{k!} W_t^k \prod_{j=0}^{k-1} \left( d_0 + 2j \right)  \|A\|_\majone \\
            & \le \frac{1}{k!} (\lambda |t| W_t)^k (d_0 + 2k - 2)^k \|A\|_\majone \\
            & \le \frac{1}{k!} (\lambda |t| W_t)^k e^{d_0 + 2k-2} k! \|A\|_\majone \\
            & = (e^2 \lambda |t| W_t )^k e^{d_0-2} \|A\|_\majone.
        \end{aligned}
    \end{equation}
    In the second-to-last line, we used that $x^k \le k! e^x$ when $x\ge 0$.
\end{proof}

\section{Efficient sampling with geometric locality}\label{sec:algorithm}

In this section we present an efficient randomized sampling algorithm for computing $\Tr(\rho_0e^{iHt } A e^{-iHt })$. Throughout this section we assume that $H$ is $(r_0, \frakd)$-geometrically local, and the fermions are arranged on a $D$-dimensional lattice. Instead of computing the integrand in \eqref{eq:integrand_kth_order} directly using linear algebra operations within the truncated subspace, the algorithm keeps track of only one Majorana monomial, which is computationally cheaper.
We present a method that, given input $(A, k, t)$, generates unbiased samples $(\omega, \nu)$ satisfying $\E (\omega \nu) = \widetilde{A}^{(k)}(t)$, where $\omega$ is a signed weight and $\nu$ a Majorana monomial.

Consider the action of $i[\tau^0_{-t}(V), \cdot] $ on a Majorana monomial. The resulting operator is a linear combination of Majorana monomials. Our algorithm samples one of these monomials with probability proportional to the magnitude of its coefficient.
More generally, for $X=\sum_{S\subseteq[2N]} a_S\gamma_S$, by \emph{sampling} from \(X\), we mean sampling the monomial \(\gamma_S\) with probability
\(p(S)=\frac{|a_S|}{\|X\|_{\majone}}\).
This can be implemented efficiently when \(X\) contains only polynomially many terms in \(N\).
We also record the signed sampling weight \(\omega = \sgn(a_S)\|X\|_\majone\). If \(\nu\) is sampled from \(X\) according to the above distribution, then $\E (\omega \nu ) = X$,
so the resulting weighted sample is unbiased.
Given an input Majorana monomial $\nu$ and signed weight $\omega$, we refer to this procedure of sampling from $i[\tau^0_{-t}(V), \nu]$ and updating the signed weight $\omega$ by \textsc{SampleNext}$(\omega, \nu, t)$; it is described in detail later.

Our algorithm \textsc{Sample}$(A, k, t)$ generates unbiased samples for $\widetilde{A}^{(k)}(t)$
by iteratively applying \textsc{SampleNext}.
Without loss of generality, we assume $t>0$. We first focus on the case where $A$ is a single Majorana monomial.
When $A$ is a sum of Majorana monomials, we apply the same procedure to each term and add the results.
Based on \eqref{eq:kth_order}, which we restate here
\begin{equation*}
    \widetilde{A}^{(k)}(t) = (i\lambda)^k \int_0^{t} \dd t_k \int_0^{t_{k}} \dd t_{k-1} \cdots \int_0^{t_{2}} \dd t_1 \, \Big[ \tau^0_{-t_k} (V), \big[\tau^0_{-t_{k-1}}(V), \cdots, [\tau^0_{-t_1}(V), A]\big] \Big],
\end{equation*}
we first sample  tuple $(t_1,\ldots,t_k)$ uniformly distributed over the
time-ordered simplex,  $0\le t_1\le t_2\le \cdots \le t_k\le t$, which can be done as described in the proof of Corollary 
\ref{cor:quasi-polynomial-time}.

We then initialize $\omega$ to be $1$ and $\nu$ to be $A$. Next, we go through the $k$ layers of commutators one by one.
At the $q$th layer, we update $(\omega, \nu)$ using the subroutine \textsc{SampleNext}$(\omega, \nu, t_q)$.
Finally, after going through all the layers,
we multiply the weight by $\lambda^k t^k / k!$ which accounts for the integration and the prefactor.
Because each call to \textsc{SampleNext} produces an unbiased sample, the final sample is unbiased by construction and satisfies
$\E (\omega \nu) = \widetilde{A}^{(k)}(t)$.
The algorithm is described in pseudocode in Algorithm~\ref{alg:sample}.

\begin{figure}[htbp!]
    \centering

    \vspace{-2cm}

    \begin{minipage}{\textwidth}
        \centering
        \includegraphics[width=0.3\linewidth]{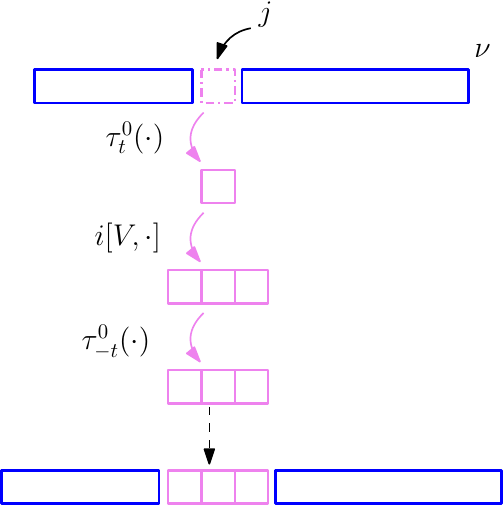}
        \caption{An illustration of one iteration of \textsc{SampleNext}, 
         supposing the input is $(\omega, \nu, t)$. 
         The updating of the weight $\omega$ is omitted in the above diagram.
         First, one Majorana mode $j$ is chosen randomly from $\nu$. Then, each violet arrow means to compute the action of a 
         superoperator on the previously sampled Majorana monomial
         and then sample from the result. The final black dashed arrow means to multiply the 
         final sampled monomial
         by the remaining modes from the input $\nu$. 
         }
        \label{fig:sampling_illustration}
    \end{minipage}

    \begin{algorithm}[H]
    \caption{\textsc{Sample}$(A,k,t)$}
    \label{alg:sample}
    \begin{algorithmic}[1]
        \REQUIRE Initial Majorana monomial $A$, perturbation order $k$,
        and evolution time $t$
        \ENSURE A signed weight $\omega$ and a Majorana monomial $\nu$
    
        \STATE Initialize $(\omega,\nu) \leftarrow (1,A)$ 
        \STATE Sample $(t_1,\ldots,t_k)$ uniformly from the time-ordered simplex
        $
            \left\{(t_1,\ldots,t_k):
            0\leq t_1\leq t_2\leq\cdots\leq t_k\leq t\right\}
        $

        \FOR{$q=1,\ldots,k$}
            \STATE $(\omega,\nu)
            \leftarrow
            \textsc{SampleNext}(\omega,\nu,t_q)$
            \IF{$\omega = 0$}
            \RETURN $(0,0)$
            \ENDIF
        \ENDFOR
    
        \STATE $\omega \leftarrow \lambda^k t^k\omega / k!$ 
        \RETURN $(\omega,\nu)$
    \end{algorithmic}
    \end{algorithm}

        \begin{algorithm}[H]
        \caption{\textsc{SampleNext}$(\omega,\nu,t)$}
        \label{alg:sample-next}
        \begin{algorithmic}[1]
            \REQUIRE Signed weight $\omega$, Majorana monomial
            $\nu=\frakI(\nu)\nu_1\cdots\nu_J$, and time $t$
            \ENSURE Updated signed weight $\omega$ and Majorana monomial $\nu$
        
            \STATE Sample $j$ uniformly from $[J]$
            \STATE $\omega \leftarrow \omega\,J$ 
            \STATE Sample $\nu'$ from $\tau_t^0(\nu_j)$
            \STATE
            $
                \omega \leftarrow
                \omega\,
                \bigl\|\tau_t^0(\nu_j)\bigr\|_{\majone}\,
                \sgn\langle\nu',\tau_t^0(\nu_j)\rangle
            $
        
            \IF{$[V,\nu']=0$}
            \RETURN $(0,0)$
            \ENDIF
            
            \STATE Sample $\nu''$ from $i[V,\nu']$
        
            \STATE
            $
                \omega \leftarrow
                \omega\,
                \bigl\|[V,\nu']\bigr\|_{\majone}\,
                \sgn\langle\nu'',i[V,\nu']\rangle
            $
            \STATE $\nu' \leftarrow \nu''$
        
            \STATE Sample $\nu''$ from $\tau_{-t}^0(\nu')$ 
            \STATE
            $
                \omega \leftarrow
                \omega\,
                \bigl\|\tau_{-t}^0(\nu')\bigr\|_{\majone}\,
                \sgn\langle\nu'',\tau_{-t}^0(\nu')\rangle
            $
            \STATE $\nu' \leftarrow \nu''$
        
            \STATE
            $
                S \leftarrow
                \majsupp(\nu')
                \mathbin{\Delta}
                \majsupp\!\left(
                    \nu_1\cdots\nu_{j-1}\nu_{j+1}\cdots\nu_J
                \right)
            $
        
            \STATE
            $
                \mu \leftarrow
                \frakI(\nu)
                \left(\prod_{m=1}^{j-1}\nu_m\right)
                \nu'
                \left(\prod_{m'=j+1}^{J}\nu_{m'}\right)
            $
        
            \STATE $\omega \leftarrow \omega\,\sgn \langle\gamma_S,\mu \rangle$ 
            \STATE $\nu \leftarrow \gamma_S$
        
            \RETURN $(\omega,\nu)$
        \end{algorithmic}
        \end{algorithm}

\end{figure}

We now formally introduce \textsc{SampleNext},
which takes input $(\omega, \nu, t)$, samples a new Majorana monomial from $i[\tau^0_{-t} V, \nu]$, updates $\omega$, and returns the new $(\omega, \nu)$. The pseudocode is in Algorithm \ref{alg:sample-next}, with an illustrative diagram in Figure \ref{fig:sampling_illustration}. Sampling from $i[\tau^0_{-t} V, \nu]$ is based on the decomposition \eqref{eq:commutator_expansion}, which we restate here for ease of reference (multiplied by $i$ on both sides):
\begin{equation}
    \begin{aligned}
        i[\tau^0_{-t} V, \nu] & = \frakI (\nu) \sum_{j=1}^{J} \Big( \prod_{m=1}^{j-1} \nu_m \Big)  \,  i[ \tau^0_{-t} V, \nu_{j} ] \, \Big( \prod_{m'=j+1}^{J} \nu_{m'} \Big) , \\
        & = \frakI (\nu) \sum_{j=1}^{J} \Big( \prod_{m=1}^{j-1} \nu_m \Big)  \,  \tau^0_{-t} ( i [  V, \tau^0_{t}(\nu_{j}) ] ) \, \Big( \prod_{m'=j+1}^{J} \nu_{m'} \Big).
    \end{aligned}
\end{equation}
As suggested by this equation, we first sample a Majorana operator component $\gamma_j$ uniformly from $\nu$ (Line $1$-$2$ in Algorithm \ref{alg:sample-next}).
Then we update this component by sampling from $i[\tau^0_{-t} V, \gamma_j]$. Since $i[\tau^0_{-t} V, \gamma_j] = \tau^0_{-t} (i [  V, \tau^0_{t}(\nu_{j}) ])$, we do the following three-step procedure: (1) Sample from $\tau^0_t(\nu_j)$ (Line $3$-$4$ in Algorithm \ref{alg:sample-next}). Here, $\nu_j$ is a degree-$1$ monomial.
(2) Sample from $i[V, \nu']$, where we denote the sampled Majorana monomial from the last step by $\nu'$ (Line $5$-$10$ in Algorithm \ref{alg:sample-next}). Here, $\nu'$ is degree-$1$. In particular if the commutator is $0$, we terminate the algorithm and return a zero sample. 
(3) Sample from $\tau^0_{-t}(\nu')$, where we used $\nu'$ again to denote the sampled Majorana monomial from the last step (Line $11$-$13$ in Algorithm \ref{alg:sample-next}). Here, $\nu'$ is degree-$3$.
Finally, we multiply this sample by the remaining modes from the input and format the result in the canonical Majorana monomial form defined in \eqref{eq:Majorana_monomial} together with a signed weight (Line $14$-$18$ in Algorithm \ref{alg:sample-next}).

In the above procedure, the main sampling steps are sampling from $\tau^0_{\pm t} (\nu)$, where $\nu$ is a degree-$1$ or degree-$3$ Majorana monomial, and sampling from $i[V, \nu]$ where $\nu$ is a degree-$1$ Majorana monomial.
By \eqref{eq:htilde_def_all}, the action of $\tau^0_{\pm t} $ on a degree-$1$ Majorana polynomial can be represented by $e^{\pm i\htilde t}$. Computing $e^{\pm i\htilde t}$ can be done in $\Or(N^3)$.\footnote{Since our goal is not to obtain an optimal runtime, but only to make it polynomial, we will use the most direct method to compute $e^{i\htilde t}$, i.e., diagonalize the matrix $\htilde$ and then exponentiate the eigenvalues. This direct method takes time $\mathcal{O}(N^3)$ \cite{GuEtAlTemplates}. More sophisticated methods with better asymptotic scaling are certainly possible, but this is not the focus of this work.} To compute $\tau^0_{\pm t} (\gamma_j)$, we simply extract the $j$th column of this matrix; then we sample from it.
To compute $\tau^0_{\pm t} (\nu)$ where $\nu$ is degree-$3$, we use the multiplicativity of the free evolution: \(\tau^0_{\pm t}(\nu_1 \nu_2 \nu_3) = \tau^0_{\pm t}(\nu_1)\tau^0_{\pm t}(\nu_2)\tau^0_{\pm t}(\nu_3)\). We compute for each Majorana mode separately and multiply the resulting monomials together. Computing for each degree-$1$ component takes time $\Or(N^3)$ and the multiplication takes time $\Or(N^3)$, therefore the overall runtime is $\Or(N^3)$. By locality, only $\Or(\frakd)$ terms in $V$ overlap with $\nu$ when computing $i[V, \nu]$. Therefore, computing and sampling from $i[V, \nu]$ where $\nu$ is degree-$1$ takes time $\Or(N)$. In total, one call to \textsc{SampleNext} takes time $\Or(N^3)$.

In the end, we need to compute the expectation value of a Majorana monomial for a fermionic Gaussian state. This can be computed as the Pfaffian of the corresponding covariance submatrix \cite{surace2022fermionic}. Since computing the covariance submatrix and computing the Pfaffian both take time $\Or(N^3)$ in the worst case \cite{wimmer2012algorithm}, a single Majorana monomial expectation value takes at most $\Or(N^3)$ time to compute.
Since the runtime of \textsc{Sample}$(A, k ,t)$ is dominated by $k$ calls to \textsc{SampleNext}  plus evaluating the expectation value once , \textsc{Sample}$(A, k ,t)$ runs in time
$\OO((k+1)N^3) = \Or(kN^3)$ when $k\ge 1$ ($k=0$ is trivial).

We next bound the sampling variance, which determines the sample complexity.
Let $(\omega^{(k)}, \nu^{(k)})$ be the sample generated by $\textsc{Sample}(A, k, t)$.
As discussed earlier,
$\E (\omega^{(k)} \nu^{(k)}) = \widetilde{A}^{(k)}(t)$.
Initially, $\omega = 1, \nu = A$.
Consider iteration $q$ in \textsc{Sample}, in which \textsc{SampleNext} is called with $(\omega, \nu, t_q)$, where $(\omega, \nu)$ is the weight-monomial pair produced at iteration $q-1$ and $|t_q| \le |t|$.
The four sampling steps in \textsc{SampleNext} contribute weight factors bounded by (1) \(\deg(\nu)\), (2) \(\sup_{\deg(\nu)=1}\|\tau^0_{t_q}(\nu)\|_\majone \le \|T_{t_q} \calP_1\|_\onetoone\) where the supremum is over degree-$1$ Majorana monomials, (3) $\sup_{\deg(\nu)=1}\|i[V, \nu]\|_\majone \le \|Q\calP_1\|_\onetoone$ where the supremum is over degree-$1$ Majorana monomials, (4) \(\sup_{\deg(\nu)=3}\|\tau^0_{-t_q}(\nu)\|_\majone \le \|T_{-t_q} \calP_3\|_\onetoone\) where the supremum is over degree-$3$ Majorana monomials. By Lemma \ref{lem:commutator_matrix_norm}, \ref{lem:1norm_submulplicativity}, \ref{lem:eiht_bnd}, they are further upper bounded respectively by \(
\deg(\nu), L_t, 2\frakd, L_t^3,
\)
where $L_t =  1 + C_D \lceil 4e\frakd |t| r_0\rceil^{D/2} $.
Since the degree before iteration \(q\) is at most \(d_0+2q-2\), multiplying these bounds over the \(k\) iterations gives
\begin{equation}\label{eq:weight_upp_bnd}
    \begin{aligned}
        |\omega^{(k)}| & \le \lambda^k \frac{t^k }{k!}  (2\frakd L_t^4)^k \prod_{j=0}^{k-1} \left( d_0 + 2j \right)\\
        & \le \frac{1}{k!} (\lambda |t| W_t)^k e^{d_0 + 2k-2} k! \\
        & \le (e^2 \lambda |t| W_t )^k e^{d_0-2},
    \end{aligned}
\end{equation}
where
$W_t = 2 \frakd \left( 1 + C_D \lceil 4e\frakd |t| r_0\rceil^{D/2} \right)^4$.

Next, we consider a more general observable $A$ that may not be a single Majorana monomial.
For initial observable $A= \sum_{i=1}^M a_i O_i$, where $O_i$'s are Majorana monomials and $|a_i| \le 1$, we have
\begin{equation}
    \widetilde{A}(t) = \sum_{i=1}^M a_i \sum_{k=0}^{\infty} \E(\omega^{(k)}_i \nu^{(k)}_i),
\end{equation} where $(\omega^{(k)}_i,  \nu^{(k)}_i)$ are samples obtained from $\textsc{Sample}(O_i, k, t)$.
We will generate and average samples for each $(i,k)$ for all $i \in [M]$ and for $k$ up to some truncation order $K$.
We give the following theorem on the overall runtime.

\begin{theorem}\label{thm:sample_complexity_l1}
    Consider an
    interacting fermionic system
    \eqref{eq:interacting_fermions_intro}
    on a $D$-dimensional lattice with a $(r_0, \frakd)$-geometrically local Hamiltonian (Definition~\ref{def:geo_local}). Let $N$ be the number of fermionic modes,
    $A=\sum_{i=1}^M a_i O_i$ be an observable given as a list of Majorana monomials $\{O_i\}_i$ with coefficients $\{a_i\}_i$ satisfying $|a_i|\le 1$, and $\rho_0$ be a Gaussian state. Let
    \begin{equation*}
        W_t = 2 \frakd \left( 1 + C_D \lceil 4e \frakd |t| r_0 \rceil^{D/2}  \right)^4.
    \end{equation*}
    For $e^2 \lambda |t|W_t < 1$,
    there is a sampling algorithm that computes $\Tr(\rho_0\, e^{iHt } A e^{-iHt })$ to within additive error $2\epsilon$ with probability at least $1-\delta$ in time
    \begin{equation*}
        \OO \left( {e^{2d_0}} \left(1 - \left(e^2\lambda |t| W_t \right)^2\right)^{-1} M^3 K^3 \epsilon^{-2} \log(MK/\delta) \, N^3 \right)
    \end{equation*} where $d_0 = \deg(A)$ and
    \begin{equation*}
            K = \Theta \left( \log \left( \frac{M e^{d_0}}{ (1 - e^2 \lambda |t| W_t) \epsilon } \right) \bigg/  \log \frac{1}{e^2 \lambda |t| W_t}  \right)
    \end{equation*}
\end{theorem}

\begin{proof}
    Due to Theorem~\ref{thm:convergence_maj_one}, the truncation error is upper bounded by
    \begin{equation*}
        \sum_{k=K+1}^{\infty} Me^{d_0-2} \left( e^2 \lambda |t| W_t  \right)^k = Me^{d_0-2} \frac{\left( e^2 \lambda |t| W_t  \right)^{K+1}}{1 - e^2 \lambda |t| W_t }.
    \end{equation*}
    Therefore, we truncate at $K$ as given in the theorem, controlling the truncation error to $\epsilon$.

    We want to evaluate each term to within additive error $\epsilon/(MK)$ with probability $1 - \delta/(MK)$.
    For the $k$th order of the $i$th term, we compute scalar samples \( \Tr\left( e^{-iH^0t} \rho_0 e^{iH^0t} a_i \omega_i^{(k)} \nu_{i}^{(k)}  \right) \). Because every Majorana monomial has operator norm $1$ and $|a_i|\le 1$, the scalar samples are bounded by \(|\omega_i^{(k)}|\) in absolute value.
    \(e^{-iH^0t} \rho_0 e^{iH^0t}\) remains Gaussian, and therefore we will compute Majorana monomial expectation via computing the corresponding Pfaffian.
    Since $|\omega_i^{(k)}| \le e^{d_0-2} (e^2 \lambda |t| W_t)^k$ when $k\ge 1$, and  $|\omega_i^{(0)}| = 1$, we let $G_{i,k} = e^{d_0} (e^2 \lambda |t| W_t)^k$, which upper bounds $|\omega_i^{(k)}|$ for all $k\ge 0$.\footnote{Note that in its current form $G_{i,k}$ does not depend on $i$, because we let  the observable coefficients be bounded by the universal constant $1$. In the more general setting $G_{i,k}$ depends on $|a_i|$, so here we retain the subscript $i$ for clarity.}
    Then, by Hoeffding's inequality, the  sample complexity for the $(i,k)$th-term is
    \begin{equation*}
        \OO\left( G_{i,k}^2  \left( \frac{MK}{\epsilon} \right)^2 \log \frac{MK}{\delta} \right).
    \end{equation*}
    Since each sample for each term can be generated in time $\OO(KN^3)$ as discussed above,
    the total runtime is
    \begin{equation*}
        \OO\left( \left( \sum_{k=0}^K \sum_{i=1}^M G_{i,k}^2 \right) \cdot \left( \frac{MK}{\epsilon} \right)^2 \log \frac{MK}{\delta} \cdot K N^3 \right).
    \end{equation*}
    We upper bound $\sum_{k=0}^K \sum_{i=1}^M G_{i,k}^2$ by
    \begin{equation*}
        \sum_{k=0}^\infty \sum_{i=1}^M G_{i,k}^2 \le M e^{2d_0} \Big/ \left( 1 - \left(e^2\lambda |t| W_t \right)^2\right).
    \end{equation*}
    This gives the runtime
    \begin{equation*}
        \OO \left(
            M e^{2d_0} \Big/ \left( 1 - \left(e^2\lambda |t| W_t \right)^2\right) \cdot \left( \frac{MK}{\epsilon} \right)^2 \log \frac{MK}{\delta} \cdot K N^3
         \right)
    \end{equation*} in the theorem.
\end{proof}

\section{Application to Anderson localization}\label{sec:Anderson}

In this section, we show a tighter bound of $\|e^{i\htilde t}\|_{\onetoone}$ from the perspective of quantum transport for the Anderson model.
For concreteness, let us focus on the case where each site has a single fermionic mode.
The free Hamiltonian $H^0$ in the Anderson model consists of a kinetic energy term and a random potential of strength $\Delta > 0$
diagonal in the position basis:
\begin{equation}
    H^0 = T + \Delta V_\omega.
\end{equation}
In terms of creation and annihilation operators,
\begin{equation}\label{eq:Anderson_Hamiltonian}
    H^0 = \sum_{x, y \text{ nearest neighbors}} t_{xy} \hat{\psi}_x^\dag \hat{\psi}_y + \Delta \sum_x {\omega_x} \hat{\psi}_x^\dag \hat{\psi}_x.
\end{equation}
$\{\omega_x\}_x$ are i.i.d.\ random variables in $\mathbb{R}$ with a bounded and compactly supported density,
$(t_{x,y})$ is a Hermitian matrix satisfying $|t_{x,y}|\leq 1$.
It is known that for $D=1$, the entire spectrum is localized at any level of disorder $\Delta > 0$. For $D \ge 2$, we focus on the large-disorder regime where the entire spectrum is localized, i.e. we require $\Delta \geq \Delta_0$ for some threshold $\Delta_0=\Or(1)$ \cite[Theorem~3]{stolz2011introanderson} (also in \cite{aizenman1993localization}).

In this section, we consider an interacting Hamiltonian of the following form:
\begin{equation}
    \label{eq:interacting_anderson_ham}
    H = H^0 + \lambda \sum_{i,j,k,l=1}^{2N} v_{ijkl} \gamma_{i}\gamma_{j}\gamma_{k}\gamma_{l},
\end{equation}
with $H^0$ defined in \eqref{eq:Anderson_Hamiltonian}. $|v_{ijkl}|\le1$. We assume that the Majorana modes are arranged on a $D$-dimensional lattice and that $H$ is $(r_0, \frakd)$-geometrically local.
We note that in this section we no longer assume that the coefficient of each term in $H^0$ is bounded by $1$, but instead allow $\Delta$ to be a large constant.

Dynamical localization gives the following bound on the expected transition amplitude:
\begin{equation}
\label{eq:dynamical_localization}
    \E \left(  \sup_{t\in \mathbb{R}} \left| \braket{ e_y|  e^{-iH^0 t} |e_x  }  \right| \right) \le C e^{-\mu d(x, y)}, \quad \forall x, y \in \Gamma
\end{equation} for some constant $C < \infty$ and $\mu > 0$. Here $\ket{e_x} = \hat{\psi}_x^\dag \ket{0}$ and $\ket{e_y}=\hat{\psi}_y^\dag \ket{0}$ are eigenstates of the position operator. This bound comes as a result of combining Theorems 3 and 4 in \cite{stolz2011introanderson}.
In previous sections we worked in the Majorana basis. Note that we can readily switch between the Majorana operators and the creation and annihilation operators in \eqref{eq:dynamical_localization}. In this section, let $\htilde$ be the matrix representation of $H^0$ in \eqref{eq:Anderson_Hamiltonian}.
Suppose $\gamma_\alpha = \hat{\psi}_x^\dag + \hat{\psi}_x$ and $\gamma_\beta =  i \left( \hat{\psi}_y^\dag - \hat{\psi}_y \right) $ are two Majorana operators on site $x$ and $y$, where the modes $\alpha, \beta \in [2N]$,
\begin{equation*}
    \begin{aligned}
       \left| \braket{ \beta|  e^{i \htilde t} |\alpha  }  \right|
        & = \left| \overline{\Tr}\left( \gamma_\beta e^{iH^0 t} \gamma_\alpha e^{-iH^0 t}  \right) \right|  \\
        & \le \left| \overline{\Tr}\left( \hat{\psi}_y^\dag e^{iH^0 t} \hat{\psi}_x e^{-iH^0 t}  \right) \right|
            + \left| \overline{\Tr}\left( \hat{\psi}_ye^{iH^0 t} \hat{\psi}_x^\dag e^{-iH^0 t}  \right) \right| \\
        & \qquad + \underbrace{\left| \overline{\Tr}\left( \hat{\psi}_y^\dag e^{iH^0 t} \hat{\psi}_x^\dag e^{-iH^0 t}  \right) \right|}_{0} + \underbrace{\left| \overline{\Tr}\left( \hat{\psi}_y e^{iH^0 t} \hat{\psi}_x e^{-iH^0 t}  \right) \right|}_{0}  \\
        & = \frac{1}{2} \left| \braket{e_x| e^{iH^0 t} |e_y} \right| + \frac{1}{2} \left| \braket{e_y| e^{-iH^0 t} |e_x} \right|.
    \end{aligned}
\end{equation*}
Here, $\overline{\Tr}$ is the normalized trace. We have used the fact that $H^0$ is particle-number conserving so that two of the terms vanish, together with the change-of-basis formula
\begin{equation*}
    {\psi}_x^\dag(t) = \sum_{y} \braket{e_y | e_x(t) } \psi_{y}^\dag, \quad
    {\psi}_x(t) = \sum_{y} \braket{ e_x(t) | e_y} \psi_{y},
\end{equation*}
and the identity $\overline{\Tr}(\hat{\psi}_x^\dag \hat{\psi}_y)=\delta_{xy}/2$ at the end.
Then, it follows from \eqref{eq:dynamical_localization} that
\begin{equation} \label{eq:Majorana_localization}
    \E \left(  \sup_{t\in \mathbb{R}} \left| \braket{ \beta|  e^{i \htilde t} |\alpha  }  \right| \right) \le C e^{-\mu d(\alpha,\beta)}, \quad \forall \alpha, \beta \in [2N],
\end{equation} where the distance between two modes is induced by the distance between sites.

\eqref{eq:Majorana_localization} leads to a pointwise probability bound on the $\ell_1$ norm of $e^{i\htilde t}\ket{\alpha}$ applied to each mode $\alpha$, but it does not directly give the bound for $\| e^{i\htilde t} \|_{\onetoone}$ that we need, which is a uniform bound over all modes. In order to upgrade this pointwise bound to a uniform bound over all relevant modes, we consider a restriction of the Hamiltonian $H$ to a subregion $\calW\subseteq [2N]$, where only the terms fully supported in $\calW$ are included, and apply a union bound to this subregion. Let $H|_\calW, H^0|_\calW$ denote the restricted Hamiltonian and the restricted free Hamiltonian, respectively. Let $\htilde|_\calW$ denote the restricted matrix representation of $H^0|_\calW$. The restriction does not aid in operator spreading, in the sense that the localization bound \eqref{eq:Majorana_localization} holds for the restricted free Hamiltonian $\htilde|_\calW$ with constants independent of $\calW$ as long as $\calW$ has finite volume \cite{stolz2011introanderson,aizenman1993localization,Aizenman_2001finitevolume}. We show in the following lemma that applying the union bound introduces a logarithmic dependence on the volume in the $\onetoone$ norm.

\begin{lemma}\label{lem:volume_dependent_1to1norm}
    Let $\calW \subseteq [2N]$.
    Suppose $\htilde|_\calW$ satisfies dynamical localization \eqref{eq:Majorana_localization} with constants $C, \mu$. Then, there exists $C' = \Or(C), \mu' = \Or(\mu)$, such that
    \[
    \sup_{t\in\mathbb{R}} \left\| e^{i\htilde|_\calW t} \right\|_\onetoone \le 1 + C_D \left\lceil \frac{1}{\mu'} \log \frac{|\calW|C'}{\delta} \right\rceil^{\frac{D}{2}}
    \] with probability $1-\delta$. Here $C_D$ is the same $D$-dependent constant as in Lemma \ref{lem:eiht_bnd}.
\end{lemma}

\begin{proof}
    Let $P(\alpha, \beta; t)$ denote the transition amplitude at time $t$: $P(\alpha, \beta; t) = |\braket{ \beta|  e^{i \htilde|_\calW t} |\alpha  }|$.
    Let $P(\alpha, \beta)$ denote the supremum transition amplitude: $P(\alpha, \beta) := \sup_{t\in \mathbb{R}} |\braket{ \beta|  e^{i \htilde|_\calW t} |\alpha  }| = \sup_{t \in \mathbb{R}} P(\alpha,\beta;t)$. $P(\cdot, \cdot)$ is not necessarily stochastic.

    Let $r > 0$. First, for any fixed $t \in \mathbb{R}$ and $\alpha \in [2N]$, since $\sum_{\beta: d(\alpha, \beta) < r} P(\alpha, \beta; t)^2 \le 1$, we have $\sum_{\beta: d(\alpha, \beta) < r} P(\alpha, \beta; t) \le C_D \lceil r \rceil ^{\frac{D}{2 }}$, where $C_D>0$ is the $D$-dependent constant that satisfies $|\{\beta: d(\beta,\alpha) < l\}|\leq C_D^2 \lceil l \rceil^D$.

    We then show that
    \begin{equation}\label{eq:r_dependent_1to1norm}
        \sup_t \left\| e^{i\htilde|_\calW t} \right\|_\onetoone \le C_D \lceil r \rceil^{\frac{D}{2}} + \sup_\alpha \sum_{\beta: d(\alpha,\beta) \ge r} P(\alpha, \beta).
    \end{equation}
    For the left-hand side, we have
    \begin{equation}
        \begin{aligned}
            \sup_{t\in\mathbb{R}} \left\| e^{i\htilde|_\calW t} \right\|_\onetoone
            & = \sup_{t\in\mathbb{R}} \sup_{\alpha} \sum_{\beta} P(\alpha, \beta; t) \\
            & \le \sup_{t\in\mathbb{R}} \sup_{\alpha} \left( C_D \lceil r \rceil ^{\frac{D }{2}} + \sum_{\beta: d(\alpha, \beta) \ge r} P(\alpha, \beta; t)  \right) \\
            & = C_D \lceil r \rceil ^{\frac{D }{2}} + \sup_{\alpha} \sup_{t\in\mathbb{R}} \sum_{\beta: d(\alpha, \beta) \ge r} P(\alpha, \beta; t) \\
            & \le C_D \lceil r \rceil ^{\frac{D }{2}} + \sup_{\alpha}  \sum_{\beta: d(\alpha, \beta) \ge r} \sup_{t\in\mathbb{R}} P(\alpha, \beta; t) \\
            & = C_D \lceil r \rceil^{\frac{D}{2}} + \sup_\alpha \sum_{\beta: d(\alpha,\beta) \ge r} P(\alpha, \beta).
        \end{aligned}
    \end{equation}

    By integrating over the lattice, \eqref{eq:Majorana_localization} implies that
    \begin{equation}
       \E \left( \sum_{\beta: d(\alpha,\beta) \ge r} P(\alpha, \beta) \right)
       \le C' e^{-\mu' r}
    \end{equation} for constants $C' = \Or(C), \mu' = \Or(\mu)$. By Markov's inequality,
    \begin{equation}
        \PP \left( \sum_{\beta: d(\alpha,\beta) \ge r} P(\alpha, \beta) > 1 \right) \le C' e^{-\mu' r}.
    \end{equation}
    By a union bound,
    \begin{equation} \label{eq:union_bnd}
        \PP \left( \exists \alpha \in \calW: \sum_{\beta: d(\alpha,\beta) \ge r} P(\alpha, \beta) > 1 \right) \le C' e^{-\mu' r} |\calW|.
    \end{equation}

    We choose
    \begin{equation}
        r = \left\lceil \frac{1}{\mu'} \log \frac{|\calW|C'}{\delta} \right\rceil.
    \end{equation} Then, the probability \eqref{eq:union_bnd} is at most $\delta$.
    Combining with \eqref{eq:r_dependent_1to1norm}, we have
    \begin{equation}
        \sup_{t\in\mathbb{R}} \left\| e^{i\htilde|_\calW t} \right\|_\onetoone \le 1 + C_D \left\lceil \frac{1}{\mu'} \log \frac{|\calW|C'}{\delta} \right\rceil^{\frac{D}{2}}
    \end{equation} with probability $1-\delta$.
\end{proof}

We consider initial observable $A = \sum_{i=1}^M a_iO_i$. For a mode $\alpha\in [2N]$, let $\mathfrak{B}_{R}(\alpha)$ be the ball of modes of radius $R$ centered at $\alpha$. For a collection of modes $\Lambda \subseteq [2N]$ , let $\mathfrak{B}_{R}(\Lambda) := \cup_{\alpha\in \Lambda} \mathfrak{B}_{R}(\alpha)$.
We consider restricting the dynamics to $\mathfrak{B}_R(\supp(A))$.
For our purpose of simulating observable evolution, this introduces an error that exponentially decays with $R$ \cite{Haah2021latticehamiltonians,nachtergaele2010liebrobinsonbounds,Nachtergaele2008LRboundharmoniclattice}. On-site Hamiltonians, even with arbitrarily large coefficients, do not enter the Lieb-Robinson propagation velocity \cite{Nachtergaele2008LRboundharmoniclattice}. We focus on the weak interaction regime, assuming $\lambda < \lambda_0 = \Or(1)$. Then, there exist constants $B, J, \xi = \Or(1) $ that depend only on $\lambda_0, \frakd, r_0$ and $D$ such that when $R > J|t|$, the following holds:
\begin{equation} \label{eq:LR_bound}
        \| e^{iH t} A e^{-iH t} - e^{iH|_{\mathfrak{B}_R(\supp(A))} t} A e^{-iH|_{\mathfrak{B}_R(\supp(A))} t} \|
        \le B |\supp(A)| \|A\| e^{- \xi R}
\end{equation} where $H|_{\mathfrak{B}_R(\supp(A))}$ is the restriction of the Hamiltonian to the region $\mathfrak{B}_R(\supp(A))$. We choose
\begin{equation} \label{eq:choice_of_R}
    R = \left \lceil J|t| +  \frac{1}{\xi } \log \left( \frac{B|\supp(A)| \|A\|}{\epsilon} \right)\right \rceil
\end{equation}
to control the volume-restriction error to $\epsilon$. Substituting $\calW$ with $\mathfrak{B}_R(\supp(A))$ in Lemma \ref{lem:volume_dependent_1to1norm}, the following corollary on simulating the weakly interacting Anderson model follows by arguments similar to those in the previous sections.
When $d_0, M, \epsilon$ are regarded as constants, this corollary extends the efficiently simulable regime to \[\lambda |t| \log^{2D}(|t|) = \Or(1),\] asymptotically ($\lambda \to 0, |t| \to \infty$).

\begin{corollary}\label{cor:samplecomplexity_anderson}
    Consider an
    interacting fermionic system \eqref{eq:interacting_anderson_ham}
    on a $D$-dimensional lattice with a $(r_0, \frakd)$-geometrically local Hamiltonian (Definition~\ref{def:geo_local}).
    Suppose that $\Delta > 0$ if $D=1$ and $\Delta\geq \Delta_0$ for a threshold $\Delta_0=\Or(1)$ that depends only on $D,r_0,\frakd$
    if $D \ge 2$.  Let $N$ be the number of fermionic modes,
    $A=\sum_{i=1}^M a_i O_i$ be an observable given as a list of Majorana monomials $\{O_i\}_i$ with coefficients $\{a_i\}_i$ satisfying $|a_i|\le 1$, and $\rho_0$ be a Gaussian state.
    Let 
    \[
    W_t = 2\frakd \left( 1 + C_D \left\lceil \frac{1}{\mu'} \left(  \log \frac{C'}{\delta}  + \vartheta \right)\right\rceil^{\frac{D}{2}} \right)^4,
    \] where 
    \[
    \vartheta = \Or\left( D \log \left(  |t| + \log \frac{Md_0}{\epsilon}\right) + \log (Md_0) \right)
    \] and $\mu', C', C_D$ are constants.
    When $e^2 \lambda |t|W_t < 1$,
    Algorithm \ref{alg:sample} computes $\Tr(\rho_0\, e^{iHt } A e^{-iHt })$ to within additive error $3\epsilon$ with probability at least $1-2\delta$
    in time
    \begin{equation*}
        \OO \left( {e^{2d_0}} \left(1 - \left(e^2\lambda |t| W_t \right)^2\right)^{-1} M^3 K^3 \epsilon^{-2} \log(MK/\delta) \, N^3 \right)
    \end{equation*}
    where $d_0 = \deg(A)$ and
    \begin{equation*}
        K = \Theta \left( \log \left( \frac{M e^{d_0}}{ (1 - e^2 \lambda |t| W_t) \epsilon } \right) \bigg/  \log \frac{1}{e^2 \lambda |t| W_t}  \right)
    \end{equation*}
\end{corollary}
We note that the only difference between the time complexity in this corollary and that in the more general setting of Theorem \ref{thm:sample_complexity_l1} is the form of $W_t$. In this corollary, the dependence of $W_t$ on $|t|$ is polylogarithmic, whereas in Theorem \ref{thm:sample_complexity_l1}, $W_t$ grows polynomially with $|t|$.

\begin{proof}[Proof of Corollary \ref{cor:samplecomplexity_anderson}]   We allocate one $\epsilon$ for the volume truncation error \eqref{eq:LR_bound} and the other $2\epsilon$ for the sampling algorithm error as in Theorem \ref{thm:sample_complexity_l1}. We allocate one $\delta$ for the failure of Lemma \ref{lem:volume_dependent_1to1norm} and the other one $\delta$ for the failure of our sampling algorithm. By the assumption of the corollary, $|\supp(A)| \le Md_0$ and $|\supp(A)| \|A\| \le M^2 d_0$. By \eqref{eq:choice_of_R},
    \begin{equation}
        R = \Theta \left( |t| + \log \frac{Md_0}{\epsilon}   \right).
    \end{equation} The volume
    \begin{equation}
        |\mathfrak{B}_R(\supp(A))| = \Or(|\supp(A)| R^D) = \Or( Md_0 R^D).
    \end{equation}
    Substituting into Lemma \ref{lem:volume_dependent_1to1norm}, we obtain
    \begin{equation}
        \sup_{t\in\mathbb{R}} \left\| e^{i\htilde|_\calW t} \right\|_\onetoone \le 1 + C_D \left\lceil \frac{1}{\mu'} \log \frac{|\mathfrak{B}_R(\supp(A))|C'}{\delta} \right\rceil^{\frac{D}{2}} = 1 + C_D \left\lceil \frac{1}{\mu'} \left(  \log \frac{C'}{\delta}  + \vartheta \right)\right\rceil^{\frac{D}{2}},
    \end{equation} where
    \begin{equation}
        \vartheta = \Or\left( D \log \left(  |t| + \log \frac{Md_0}{\epsilon}\right) + \log (Md_0) \right).
    \end{equation}

    Following similar arguments to those in the previous sections,
    \begin{equation}
        \|T_{-t} Q T_t \calP_1 \|_\onetoone \le 2\frakd \left( \sup_{t\in\mathbb{R}} \left\| e^{i\htilde|_\calW t} \right\|_\onetoone \right)^4 \le 2\frakd \left( 1 + C_D \left\lceil \frac{1}{\mu'} \left(  \log \frac{C'}{\delta}  + \vartheta \right)\right\rceil^{\frac{D}{2}} \right)^4 =: W_t
    \end{equation} with probability $1-\delta$, where $T_{\pm t}, Q$ are now the matrix representations of the restricted Hamiltonians. With $W_t$ replaced, the remaining finite order truncation error analysis is the same as Theorem~\ref{thm:convergence_maj_one}, and the algorithm design and variance analysis are the same as Theorem~\ref{thm:sample_complexity_l1}.

\end{proof}

\section{Conclusion and outlook}

In this work we gave a classical algorithm for simulating the real-time dynamics of weakly interacting fermions, starting from a quadratic, exactly solvable Hamiltonian perturbed by a weak local quartic interaction. Working in the Heisenberg picture allowed us to keep the free fermion dynamics exact while treating the interaction perturbatively. This led to a Dyson expansion whose convergence was controlled by the spreading of local observables under the dynamics.

For bounded-degree local Hamiltonians, we proved convergence up to times $t=\mathcal{O}(1/\lambda)$, giving a quasi-polynomial-time algorithm for evolution time that is exponentially longer than what followed from previous Majorana-propagation bounds in \cite{facelli2026fastconvergencemajoranapropagation}. For geometrically local systems on a $D$-dimensional lattice, we designed a randomized sampling algorithm that runs in polynomial time for local observables when $\lambda |t|(1+|t|)^{2D}=\mathcal{O}(1)$. The sampling variance is rigorously controlled and does not grow with the system size in this regime.
Finally, we showed that disorder could substantially extend the efficiently simulable evolution time. When the quadratic part exhibited Anderson localization, the relevant operator spreading becomes only polylogarithmic in time, allowing efficient simulation when $\lambda |t|=\mathcal{O}(1)$, modulo polylogarithmic factors.

Our results suggest several directions for future work. First, it would be interesting to combine the present sampling algorithm with ideas from the inchworm method \cite{cohen2015taming,antipov2017currents,dong2017quantum,cai2020inchworm}. In the current algorithm, sampling is done for the entire interval of time between $0$ and $t$ as a whole. An inchworm-type recursion could instead divide the entire evolution into short time segments, and reuse information obtained at shorter times to construct longer-time estimates, potentially reducing the growth of variance and extending the accessible time scale.
Second, one could combine our locality-based approach with quantum embedding methods \cite{georges1996dynamical,knizia2012density,kretchmer2018real}. In such a hybrid scheme, the interacting degrees of freedom inside the relevant light cone could be treated explicitly, while part of the surrounding region is represented by an effective non-interacting bath. This may reduce the effective light-cone volume and improve the practicality of the algorithm in higher dimensions or at longer times.

The general methodology of this work should not be limited to weakly interacting fermions. A natural direction is to study Hamiltonians of the form
$H = H_0 + \lambda V,$
where $H_0$ is an exactly solvable integrable or commuting local Hamiltonian and $\lambda V$ is a weak local perturbation. This setting includes perturbed models that arise in quantum error correction and topological phases of matter \cite{bravyi2010topological}. Efficient real-time simulation in this regime would provide a controlled tool for studying the stability of exactly solvable phases, the spreading of initially local operators, and the onset of thermalization or transport induced by weak non-integrable perturbations.

It is also of interest to understand whether computational complexity places intrinsic limits on the time scale that can be efficiently simulated at a given interaction strength $\lambda$ and target precision $\epsilon$. If a classical algorithm could efficiently simulate the dynamics for arbitrarily long times, then, for sufficiently expressive families of interacting fermionic Hamiltonians, one could encode arbitrary polynomial-size quantum circuits into Hamiltonian time evolution, and thereby obtain an efficient classical simulation of $\mathsf{BQP}$, which is considered unlikely \cite{childs2013universal}. This suggests that the efficient simulation time $t$ is fundamentally limited by a function of $\lambda$ and $\epsilon$. Such a perspective could also clarify the relationship between our positive algorithmic results and possible quantum advantage, in a way analogous to recent discussions in the Gibbs-state setting \cite{tong2025fast,smid2025polynomial,smid2025rapid,Chen2025convergencecumulant,ramkumar2025high}.

\paragraph*{Acknowledgments.} This material is based upon work supported by the U.S. Department of Energy, Office of Science, Accelerated Research in Quantum Computing Centers, Quantum Utility through Advanced Computational Quantum Algorithms, grant no. DE-SC0025572 (C.Z. and Y.T.). We also acknowledge support from NSF QLCI grant OMA-2120757 and NSF
PHY-2046195.

\bibliographystyle{ieeetr}
\bibliography{ref}

\end{document}